\documentclass{amsart}
\usepackage{amsmath,amssymb,amsfonts}
\usepackage{graphicx}
\usepackage[alphabetic,y2k]{amsrefs}
\usepackage{pst-all}
\usepackage{pstricks-add}
\usepackage[all]{xy}

\newtheorem{theorem}{Theorem}[section]
\newtheorem{proposition}[theorem]{Proposition}
\newtheorem{corollary}[theorem]{Corollary}
\newtheorem{lemma}[theorem]{Lemma}
\newtheorem{remark}[theorem]{Remark}
\newtheorem{example}[theorem]{Example}
\newtheorem{definition}[theorem]{Definition}

\definecolor{lightgray}{gray}{0.9}

\definecolor{Lavender}{cmyk}{0,0.28,0,0}
\definecolor{SkyBlue}{cmyk}{0.50,0,0.08,0}
\definecolor{SpringGreen}{cmyk}{0.26,0,0.76,0}

\definecolor{FireBrick}{cmyk}{0,0.809,0.809,0.302}
\definecolor{Choco}{cmyk}{0.0,0.5,0.857,0.176}
\definecolor{DarkOlive}{cmyk}{0.206,0,0.561,0.58}
\definecolor{OliveGrab}{cmyk}{0.246,0,0.754,0.443}
\definecolor{Orange}{cmyk}{0,0.3529,1,0}
\definecolor{Magenta}{cmyk}{0,1,0,0}

\date{}

\begin{document}

\author{Vladimir Dragovi\'c}
\address{Mathematical Institute SANU, Kneza Mihaila 36, Belgrade,
Serbia\newline\indent Mathematical Physics Group, University of
Lisbon, Portugal} \email{vladad@mi.sanu.ac.rs}

\author{Milena Radnovi\'c}
\address{Mathematical Institute SANU, Kneza Mihaila 36, Belgrade, Serbia}
\email{milena@mi.sanu.ac.rs}

\title{Pseudo-integrable billiards and arithmetic dynamics}

\thanks{The research which lead to this paper was partially
supported by the Serbian Ministry of Education and Science (Project
no.~174020: \emph{Geometry and Topology of Manifolds and Integrable
Dynamical Systems}) and by Mathematical Physics Group of the
University of Lisbon (Project \emph{Probabilistic approach to finite
and infinite dimensional dynamical systems, PTDC/MAT/104173/2008}).
V.~D.~ is grateful to Prof.~Marcelo Viana and IMPA (Rio de Janeiro, Brazil) and M.~R.~ to Vered Rom-Kedar, the Weizmann Institute of Science (Rehovot, Israel), and the associateship scheme of \emph{The Abdus Salam} ICTP (Trieste, Italy) for their hospitality and support in various stages of work on this paper.
}


\begin{abstract}
We introduce a new class of billiard systems in the plane, with boundaries formed by finitely many arcs of confocal conics such that they contain some reflex angles.
Fundamental dynamical, topological, geometric, and arithmetic properties of such billiards are studied.
The novelty, caused by reflex angles on boundary, induces invariant leaves of higher genera and dynamical behaviour different from Liouville-Arnold’s theorem.
Its analogue is derived from the Maier theorem on measured foliations.
A local version of Poncelet theorem is formulated and necessary algebro-geometric conditions for periodicity are presented.
The connection with interval exchange transformation is established together with Keane's type conditions for minimality.
It is proved that the dynamics depends on arithmetic of rotation numbers, but not on geometry of a given confocal pencil of conics.
\end{abstract}

\maketitle
\tableofcontents

\section{Introduction}\label{sec:intro}
We introduce a new class of billiard systems in a plane, with boundary formed by finitely many arcs of confocal conics and with a finite number of such that they contain some reflex angles.

By \emph{a billiard within a given domain} we will assume here a dynamical system where a particle -- material point is moving freely inside the domain in an Euclidean plane, and reflecting reflecting absolutely elastically on the boundary (see \cite{KozTrBIL}).
This means that the trajectories are polygonal lines with vertices lying on the domain boundary, with congruent impact and reflection angles at each vertex, while the particle speed remains constant, see Figure \ref{fig:reflection}.
\begin{figure}[h]
\psset{unit=0.75}
\begin{pspicture}(-3,-2.5)(3,2.5)
  \SpecialCoor

\pswedge*[linecolor=gray!50](2.81,-.253){1}{80}{143}
\pswedge*[linecolor=gray!50](2.81,-.253){1}{198}{260}

\psline[linewidth=1.3pt, ArrowInside=->, arrowscale=1.4, linecolor=red]
	(0,1.8)
	(!8 sqrt 6.5 neg cos mul 5 sqrt 6.5 neg sin mul)
	(!5 sqrt 260 cos mul 2 sqrt 260 sin mul)

\psline[linewidth=1.3pt] 
(!2.81 0.050626 5 mul add -.253 0.35128 5 mul add)
(!2.81 0.050626 5 mul sub -.253 0.35128 5 mul sub)

\psellipse[linewidth=1.5pt](!8 sqrt 5 sqrt)

\end{pspicture}
\caption{Billiard reflection.}\label{fig:reflection}
\end{figure}
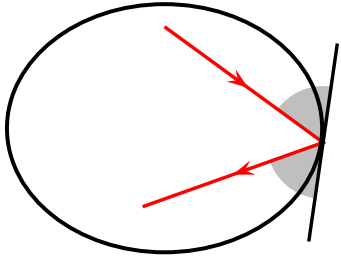

In this paper, we are going to discuss billiards in domains bounded by arcs of several confocal conics, see Figure \ref{fig:domains} for a few examples.
\begin{figure}[h]
\psset{unit=0.7}
\begin{pspicture}(-9,-2.5)(9,2.5)

% domen 1

\psellipse[linecolor=black, linewidth=1pt,fillstyle=solid, fillcolor=brown!70]
	(-6,0)(!8 sqrt 5 sqrt)
\psellipse[linecolor=black, linewidth=1pt,fillstyle=solid, fillcolor=white]
	(-6,0)(!5 sqrt 2 sqrt)

\pscircle[linestyle=none, fillstyle=solid, fillcolor=blue]
    (!3 sqrt 6 sub 0){0.1}
 \pscircle[linestyle=none, fillstyle=solid, fillcolor=blue]
    (!3 sqrt neg 6 sub 0){0.1}

% domen 2
\psellipse[linecolor=white, linewidth=0pt,fillstyle=solid, fillcolor=brown!70]
	(0,0)(!8 sqrt 5 sqrt)

\psframe[fillstyle=solid,fillcolor=white,linecolor=white](-3,-2)(-1.8,2)
\parametricplot[plotstyle=curve,linecolor=black,fillstyle=solid,fillcolor=white]
{-2.3}{2.3}{1 t t mul 2 div add sqrt neg t}

\psellipse[linecolor=black, linewidth=1pt,fillstyle=none]
	(0,0)(!8 sqrt 5 sqrt)

\pscircle[linestyle=none, fillstyle=solid, fillcolor=blue]
    (!3 sqrt 0){0.1}
 \pscircle[linestyle=none, fillstyle=solid, fillcolor=blue]
    (!3 sqrt neg 0){0.1}

% domen 3

\psellipse[linecolor=white, linewidth=0pt,fillstyle=solid, fillcolor=brown!70]
	(6,0)(!8 sqrt 5 sqrt)

\psframe[fillstyle=solid,fillcolor=white,linecolor=white](3,-2)(4.2,2)
\parametricplot[plotstyle=curve,linecolor=black,fillstyle=solid,fillcolor=white]
{-2.3}{2.3}{1 t t mul 2 div add sqrt neg 6 add t}

\psellipse[linecolor=black, linewidth=1pt,fillstyle=none]
	(6,0)(!8 sqrt 5 sqrt)

\psellipse[linecolor=brown!70, linewidth=0pt,fillstyle=solid, fillcolor=brown!70]
	(6,0)(!5 sqrt 2 sqrt)

\parametricplot[plotstyle=curve,linecolor=black,fillstyle=none]
{125}{235}{t cos 5 sqrt mul 6 add t sin 2 sqrt mul}

\pscircle[linestyle=none, fillstyle=solid, fillcolor=blue]
    (!3 sqrt 6 add 0){0.1}
 \pscircle[linestyle=none, fillstyle=solid, fillcolor=blue]
    (!3 sqrt neg 6 add 0){0.1}

\rput(-8.8,-1.8){a)}
\rput(-2.8,-1.8){b)}
\rput(3.2,-1.8){c)}

\end{pspicture}
\caption{Some domains bounded by arcs of confocal conics.}\label{fig:domains}
\end{figure}
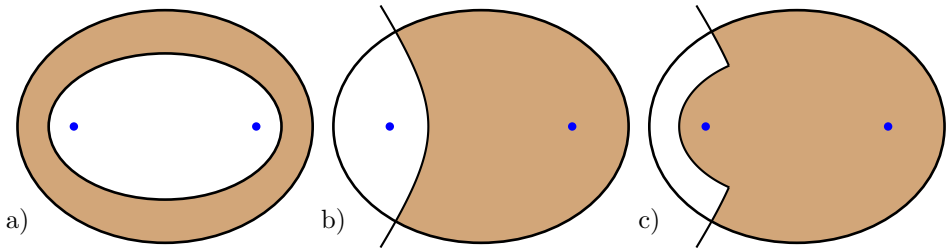

Boundaries of such domains may be non-smooth at isolated points, as it is the case for the domains shown in Figures \ref{fig:domains}b and \ref{fig:domains}c.
Since there is no tangent at such points, the billiard reflection cannot be defined there in the usual way.
However, note that two intersecting confocal quadrics are always orthogonal to each other.
If the two tangents at the meeting parts of the boundary at such a point form the convex right angle, then the reflection can be naturaly defined so that impact and reflecting segments coincide.
This definition is due to limit applied to nearby trajectories, see 
Figure \ref{fig:reflection.right}.
Notice that, because of this limit, it is natural also to count the reflection at such a point as two bounces.
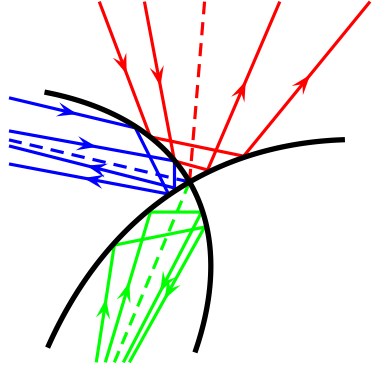
\begin{figure}[h]
\psset{unit=4}
\begin{pspicture}(-0.7,-0.7)(0.7,0.7)
  \SpecialCoor

\psset{linecolor=red,linewidth=1.2pt}
\psline[ArrowInside=->,arrowscale=1.4](-0.3,0.6)(-0.13,0.15)
\psline(-0.13,0.15)(0.18,0.086)
\psline[ArrowInside=->,arrowscale=1.4](0.18,0.086)(0.6,0.6)

\psline[ArrowInside=->,arrowscale=1.4](-0.15,0.6)(-0.05,0.07)
\psline(-0.05,0.07)(0.06,0.04)
\psline[ArrowInside=->,arrowscale=1.4](0.06,0.04)(0.3,0.6)

\psline[linestyle=dashed](0,0)(0.05,0.6)

\psset{linecolor=blue}
\psline[ArrowInside=->,arrowscale=1.4](-0.6,0.28)(-0.18,0.18)
\psline(-0.18,0.18)(-0.07,-0.04)
\psline[ArrowInside=->,arrowscale=1.4](-0.07,-0.04)(-0.6,0.06)

\psline[ArrowInside=->,arrowscale=1.4](-0.6,0.17)(-0.05,0.07)
\psline(-0.05,0.07)(-0.05,-0.02)
\psline[ArrowInside=->,arrowscale=1.4](-0.05,-0.02)(-0.6,0.12)

\psline[linestyle=dashed](0,0)(-0.6,0.14)

\psset{linecolor=green}
\psline[ArrowInside=->,arrowscale=1.4](-0.31,-0.6)(-0.25,-0.21)
\psline(-0.25,-0.21)(0.05,-0.15)
\psline[ArrowInside=->,arrowscale=1.4](0.05,-0.15)(-0.2,-0.6)

\psline[ArrowInside=->,arrowscale=1.4](-0.28,-0.6)(-0.13,-0.1)
\psline(-0.13,-0.1)(0.04,-0.1)
\psline[ArrowInside=->,arrowscale=1.4](0.04,-0.1)(-0.22,-0.6)

\psline[linestyle=dashed](0,0)(-0.25,-0.6)

\psset{linecolor=black,linewidth=2pt}

\pscustom{
	\rotate{30}
	\parametricplot[plotstyle=curve,fillstyle=none]
		{75}{110}{2 t cos mul 4 t sin mul 4 sub}
}

\pscustom{
	\rotate{30}
	\parametricplot[plotstyle=curve,fillstyle=none]
	{-30}{30}{2 t cos mul 2 sub 1 t sin mul}
}

\end{pspicture}
\caption{Reflections in right angles.}\label{fig:reflection.right}
\end{figure}

However, if the tangents form a reflex angle, the limit does not exist, thus the reflection cannot be defined, see Figure \ref{fig:reflection.concave}.
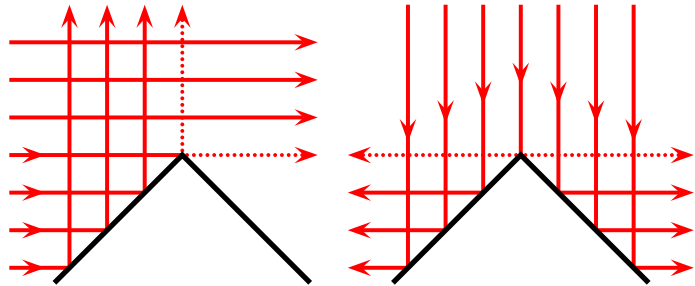
\begin{figure}[h]
\begin{pspicture}(-4.5,-2)(5,2)

\psset{linecolor=red,linewidth=1.5pt,arrowscale=1.4,dotsep=1pt}
\psline[arrows=->](-4.3,.5)(-0.2,.5)
\psline[arrows=->](-4.3,1)(-0.2,1)
\psline[arrows=->](-4.3,1.5)(-0.2,1.5)

\psline[ArrowInside=->](-4.3,0)(-3.5,0)
\psline(-3.5,0)(-2,0)
\psline[arrows=->,linestyle=dotted](-2,0)(-2,2)
\psline[arrows=->,linestyle=dotted](-2,0)(-0.2,0)

\psline[ArrowInside=->](-4.3,-0.5)(-3.5,-0.5)
\psline(-3.5,-0.5)(-2.5,-0.5)
\psline[arrows=->](-2.5,-0.5)(-2.5,2)

\psline[ArrowInside=->](-4.3,-1)(-3.5,-1)%(-3,1.8)
\psline(-3.5,-1)(-3,-1)
\psline[arrows=->](-3,-1)(-3,2)

\psline[ArrowInside=->](-4.3,-1.5)(-3.5,-1.5)
\psline[arrows=->](-3.5,-1.5)(-3.5,2)

\psline[ArrowInside=->](2.5,2)(2.5,0)
\psline[arrows=->,linestyle=dotted](2.5,0)(0.2,0)
\psline[arrows=->,linestyle=dotted](2.5,0)(4.8,0)

\psline[ArrowInside=->](3,2)(3,-0.5)
\psline[arrows=->](3,-0.5)(4.8,-0.5)

\psline[ArrowInside=->](3.5,2)(3.5,-1)
\psline[arrows=->](3.5,-1)(4.8,-1)

\psline[ArrowInside=->](4,2)(4,-1.5)
\psline[arrows=->](4,-1.5)(4.8,-1.5)

\psline[ArrowInside=->](2,2)(2,-0.5)
\psline[arrows=->](2,-0.5)(0.2,-0.5)

\psline[ArrowInside=->](1.5,2)(1.5,-1)
\psline[arrows=->](1.5,-1)(.2,-1)

\psline[ArrowInside=->](1,2)(1,-1.5)
\psline[arrows=->](1,-1.5)(.2,-1.5)

\psset{linewidth=2pt,linecolor=black}
\psline(0.8,-1.7)(2.5,0)(4.2,-1.7)
\psline(-0.3,-1.7)(-2,0)(-3.7,-1.7)

\end{pspicture}
\caption{Reflection near reflex angle.}\label{fig:reflection.concave}
\end{figure}
In the study of billiards within domains having reflex angles on the boundary (see \cite{ZorichFLAT}), of special interest are trajectories starting or being ended at the vertex of such an angle.
Such trajectories are called \emph{separatrices}.
A separatrix having both endpoins at vertices of reflex angles is called \emph{a saddle-connections}.
A saddle-connection with coinciding endpoints is called \emph{a homoclinic loop}.
All other billiard trajectories, those that never reach the vertex of an reflex angle, are called \emph{regular trajectories}.

Billiards in domains bounded by several confocal quadrics, without singular points where tangents form a reflex angle, were already studied by the authors: their periodic trajectories are described in \cites{DragRadn2004,DragRadn2006b} while their topological properties are discussed in \cite{DragRadn2009}.

In this work, we are focused to domains with reflex angles on the boundary.
We study fundamental dynamical, topological, geometric, and arithmetic properties of the corresponding billiards.
In the next section we prove the existence of a pair of independent Poisson commuting integrals.
The novelty of our systems, caused by reflex angles on the boundary, induces invariant leaves of higher genera -- see Propositions \ref{prop:genus3} and \ref{prop:N(M)}.
Its dynamical behaviour is different from Liouville-Arnold's theorem, see examples in Section \ref{sec:examples}.
An analogue of the Liouville-Arnold theorem is derived from the Maier theorem on measured foliations, see Theorem \ref{th:maier}.
A local version of Poncelet porism is formulated as Theorem \ref{th:poncelet} and necessary algebra-geometric conditions for periodicity are presented in Theorem \ref{th:cayley}.
A connection with interval exchange transformations is established in Section \ref{sec:interval} and it is proved that the dynamics depends on the arithmetic of the rotation numbers, but not on the geometry of a given confocal pencil of conics, see Theorem \ref{th:nezavisnost}.
In Section \ref{sec:keane}, we derive Keane's type conditions for minimality for interval exchange transformations that appear in such billiard systems.

\section{Elliptical billiards and confocal conics}\label{sec:ellbilliards}
Let us consider in this section billiards within an ellipse.

A general family of confocal conics in the plane can be represented in the following way:
\begin{equation}\label{eq:confocal}
\mathcal{C}_{\lambda}\ :\ \frac{x^2}{a-\lambda}+\frac{y^2}{b-\lambda}=1,
\quad\lambda\in\mathbf{R},
\end{equation}
with $a>b>0$ being constants.

By the famous \emph{Chasles' theorem} \cite{Chasles}, each segment of a given billiard trajectory is tangent to a fixed conic that is confocal to the boundary (see also
\cites{KozTrBIL,DragRadn2011book}).
This conic is called \emph{caustic} of the given trajectory.

Now, fix a constant $\alpha_0<b$ and consider billiard trajectories within confocal ellipses $\mathcal{C}_{\lambda}$ ($\lambda<\alpha_0$) having ellipse 
$\mathcal{C}_{\alpha_0}$ as caustic.

\begin{proposition}\label{prop:rotation}
There exist metrics $\mu$ on conic $\mathcal{C}_{\alpha_0}$ and function 
$$
\rho:(-\infty,\alpha_0)\to\mathbf{R}
$$
satisfying:
\begin{itemize}
\item
$\mu(\mathcal{C}_{\alpha_0})=1$;

\item
metric $\mu$ is non-atomic, i.e.\ $\mu(\{X\})=0$ for each point $X$ on 
$\mathcal{C}_{\alpha_0}$;

\item
$\mu(\ell)\neq0$ for each open arc $\ell$ of $\mathcal{C}_{\alpha_0}$;

\item
for any $\lambda<\alpha_0$, and each triplet of points
$X\in\mathcal{C}_{\alpha_0}$, $Y\in\mathcal{C}_{\alpha_0}$, $A\in\mathcal{C}_{\lambda}$, such that segments $XA$ and $AY$ satisfy the reflection law on $\mathcal{C}_{\lambda}$, the following equality holds:
$$
\mu(XY)=\rho(\lambda).
$$
\end{itemize}
\end{proposition}

\begin{proof}
Take $\lambda_0$ such that there is a closed billiard trajectory in 
$\mathcal{C}_{\lambda_0}$ with caustic $\mathcal{C}_{\alpha_0}$.
By \cite{King1994}, there is a metric $\mu$ satisfying the requested properties for 
$\lambda=\lambda_0$ -- moreover, such a metric is unique up to multiplication by a constant.
By this uniqueness property and Darboux theorem on grids \cite{DarbouxSUR} (see also 
\cite{DragRadn2006,DragRadn2008,DragRadn2011book}), it follows that metric $\mu$ satisfies the properties for each $\mathcal{C}_{\lambda}$ having closed billiard trajectories with caustic $\mathcal{C}_{\alpha_0}$.

For a periodic trajectory which becomes closed after $n$ bounces on 
$\mathcal{C}_{\lambda}$ and $m$ windings about $\mathcal{C}_{\lambda_0}$,
$\rho(\lambda)=\dfrac{m}{n}$.
Since rational numbers are dense in the reals, $\mu$ will have the required properties for all $\lambda<\alpha_0$.
\end{proof}

\begin{remark}
The function $\rho$ from Proposition \ref{prop:rotation} is called \emph{the rotation function} and its values \emph{the rotation numbers}.
Note that $\rho$ is a continously strictly decreasing function with
$\left(0,\dfrac12\right)$ as image:
$$
\lim_{\lambda\to-\infty}\rho(\lambda)=\dfrac12,
\quad
\lim_{\lambda\to\alpha_0}\rho(\lambda)=0.
$$
\end{remark}

\subsection{Elliptical billiard as a Hamiltonian system}

The standard Poisson bracket for the billiard system is defined as:
$$
\{f,g\}=\frac{\partial f}{\partial x}\frac{\partial g}{\partial\dot{x}}-
\frac{\partial f}{\partial\dot{x}}\frac{\partial g}{\partial{x}}+
\frac{\partial f}{\partial y}\frac{\partial g}{\partial\dot{y}}-
\frac{\partial f}{\partial\dot{y}}\frac{\partial g}{\partial{y}}.
$$

Define the following functions:
$$
K_{\lambda}(x,y,\dot{x},\dot{y})=\frac{\dot{x}^2}{a-\lambda}+
\frac{\dot{y}^2}{b-\lambda}-\frac{(\dot{x}y-\dot{y}x)^2}{(a-\lambda)(b-\lambda)}.
$$

\begin{proposition}
Each two functions $K_{\lambda}$ commute:
$$
\{K_{\lambda_1},K_{\lambda_2}\}=0
$$
and for $\lambda_1\neq\lambda_2$, they are functionally independent.
\end{proposition}

It is straightforward to prove the following
\begin{proposition}
Along a billiard trajectory within any conic $\mathcal{C}_{\lambda_0}$, with caustic 
$\mathcal{C}_{\alpha_0}$ and the speed of the billiard particle being equal to $s$, the value of each function $K_{\lambda}$ is constant and equal to
$$
K_{\lambda}=\frac{\alpha_0-\lambda}{(a-\lambda)(b-\lambda)}\cdot s^2.
$$
\end{proposition}

\begin{corollary}
Each $K_{\lambda}$ is integral for the billiard motion in any domain with border composed of a few arcs of confocal conics.
\end{corollary}

\section{Billiards in domains bounded by a few confocal conics}\label{sec:uvod-nosonja}
As we have already said, the aim of this paper is to analyze billiard dynamics in a domain bounded by arcs of a few confocal conics.
In order to describe some phenomena appearing in such systems, let us consider the domain $\mathcal{D}_0$ bounded by two confocal ellipses from family 
(\ref{eq:confocal}) and two segments placed on the smaller axis of theirs, as shown in Figure \ref{fig:billiard1}.
\begin{figure}[h]
\begin{pspicture}(-3.5,-2.5)(3.5,2.5)
  \SpecialCoor

\psellipticarc[linecolor=black, linewidth=1pt]
    (0,0)(!8 sqrt 5 sqrt){-90}{90}

\psellipticarc[linecolor=black, linewidth=1pt]
    (0,0)(!5 sqrt 2 sqrt){90}{270}

\psline[linewidth=1pt](!0 2 sqrt)(!0 5 sqrt)
\psline[linewidth=1pt](!0 2 sqrt neg)(!0 5 sqrt neg)

 \pscircle[linestyle=none, fillstyle=solid, fillcolor=blue]
    (!3 sqrt 0){0.1}
 \pscircle[linestyle=none, fillstyle=solid, fillcolor=blue]
    (!3 sqrt neg 0){0.1}

\rput(2.7,1.5){$\Gamma_1$}
\rput(-2.3,0.8){$\Gamma_2$}
\rput(-0.3,1.8){$\Gamma_3$}
\rput(-0.3,-1.8){$\Gamma_4$}

\end{pspicture}
\caption{Domain bounded by two confocal ellipses and two
segments on the $y$-axis.}\label{fig:billiard1}
\end{figure}
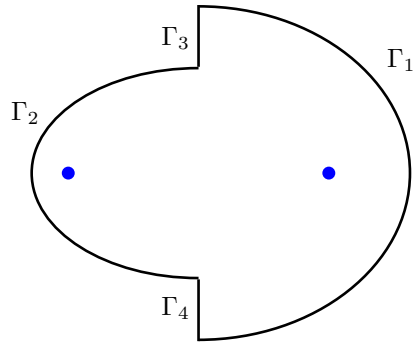
More precisely, we fix parameters $\beta_1$, $\beta_2$ such that
$\beta_1<\beta_2<b$,
and take the border of $\mathcal{D}_0$ to be:
\begin{gather*}
\partial\mathcal{D}_0=\Gamma_1\cup\Gamma_2\cup\Gamma_3\cup\Gamma_4,
\\
\Gamma_1=\{(x,y)\in\mathcal{C}_{\beta_1}\mid x\ge0\},\\
\Gamma_2=\{(x,y)\in\mathcal{C}_{\beta_2}\mid x\le0\},\\
\Gamma_3=\{(0,y)\mid\sqrt{b-\beta_2}\le y\le\sqrt{b-\beta_1}\},\\
\Gamma_4=\{(0,y)\mid-\sqrt{b-\beta_1}\le y\le-\sqrt{b-\beta_2}\}.
\end{gather*}
Notice that segments $\Gamma_3$, $\Gamma_4$ are lying on the the degenerate conic 
$\mathcal{C}_a$ of family (\ref{eq:confocal}).

By Chasles' theorem \cite{Chasles}, each line in the plane is touching exactly one conic from a given confocal family -- moreover, this conic remains the same after the reflection on any conic from the family.
Thus, each billiard trajectory in a domain bounded by arcs of several confocal conics has a caustic from the confocal family.

Consider billiard trajectories within domain $\mathcal{D}_0$ whose caustic is an ellipse 
$\mathcal{C}_{\alpha_0}$
completely placed inside the billiard table, i.e.~ $\beta_2<\alpha_0<b$.
An example of such a trajectory is shown in Figure \ref{fig:caustic1}.
\begin{figure}[h]
\begin{pspicture}(-3,-2.5)(3,2.5)
  \SpecialCoor
%  \psellipse(!8 sqrt 5 sqrt)
%  \psellipse(!8 3 sub sqrt 5 3 sub sqrt)
%  \psellipse(!8 4 sub sqrt 5 4 sub sqrt)

\psline[linewidth=1pt, linecolor=red]
      (!8 sqrt 19 neg cos mul 5 sqrt 19 neg sin mul)
      (!8 sqrt 79 cos mul 5 sqrt 79 sin mul)
      (0,1.8)
      (!8 sqrt 6.5 neg cos mul 5 sqrt 6.5 neg sin mul)
      (!5 sqrt 260 cos mul 2 sqrt 260 sin mul)
      (!5 sqrt 189 cos mul 2 sqrt 189 sin mul)

\psellipticarc[linecolor=black, linewidth=1pt]
    (0,0)(!8 sqrt 5 sqrt){-90}{90}

\psellipticarc[linecolor=black, linewidth=1pt]
    (0,0)(!5 sqrt 2 sqrt){90}{270}

\psellipse[linecolor=black, linewidth=0.8pt, linestyle=dashed]%
    (0,0)(!4 sqrt 1)

\psline[linewidth=1pt](!0 2 sqrt)(!0 5 sqrt)
\psline[linewidth=1pt](!0 2 sqrt neg)(!0 5 sqrt neg)

 \pscircle[linestyle=none, fillstyle=solid, fillcolor=blue]
    (!3 sqrt 0){0.1}
 \pscircle[linestyle=none, fillstyle=solid, fillcolor=blue]
    (!3 sqrt neg 0){0.1}
\end{pspicture}

%\parametricplot[linecolor=red, linewidth=1pt, plotpoints=10000]{0}{360}
%   {
%    /f 30 def
%    /a 6.5 sqrt def
%    /b 2.3 sqrt def
%    a t cos f cos mul mul  b t sin f sin mul mul  add
%    a t cos f sin mul mul  b t sin f cos mul mul  sub
%   }
\caption{A billiard trajectory in $\mathcal{D}_0$ with an ellipse as caustic.}\label{fig:caustic1}
\end{figure}
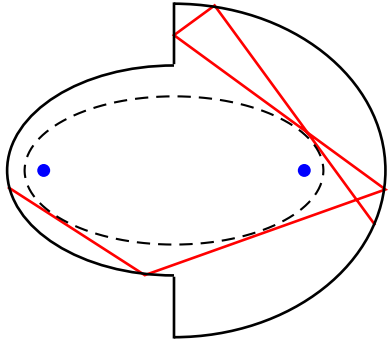

Such billiard trajectories fill out the ring $\mathcal R$ placed
between the billiard border and the caustic, see Figure
\ref{fig:ring1}.

\begin{figure}[h]
\begin{pspicture}(-3,-2.5)(4,2.5)
  \SpecialCoor
%  \psellipse(!8 sqrt 5 sqrt)
%  \psellipse(!8 3 sub sqrt 5 3 sub sqrt)
%  \psellipse(!8 4 sub sqrt 5 4 sub sqrt)

\psellipticarc[linecolor=black, linewidth=1pt, fillstyle=solid, fillcolor=blue!50]%
    (0,0)(!8 sqrt 5 sqrt){-90}{90}

\psellipse[linestyle=none, fillstyle=solid, fillcolor=blue!50]%
    (0,0)(!5 sqrt 2 sqrt)

\psellipticarc[linecolor=black, linewidth=1pt]%
    (0,0)(!5 sqrt 2 sqrt){90}{270}

\psellipse[linecolor=black, linewidth=1pt, fillstyle=solid, fillcolor=white]%
    (0,0)(2, 1)

\psline[linewidth=1pt](!0 2 sqrt)(!0 5 sqrt)
\psline[linewidth=1pt](!0 2 sqrt neg)(!0 5 sqrt neg)

\rput(2.7,1.5){$\Gamma_1$}
\rput(-2.3,0.8){$\Gamma_2$}
\rput(-0.3,1.8){$\Gamma_3$}
\rput(-0.3,-1.8){$\Gamma_4$}

\rput(1.7, 0){$\mathcal{C}$}%

\end{pspicture}

%\parametricplot[linecolor=red, linewidth=1pt, plotpoints=10000]{0}{360}
%   {
%    /f 30 def
%    /a 6.5 sqrt def
%    /b 2.3 sqrt def
%    a t cos f cos mul mul  b t sin f sin mul mul  add
%    a t cos f sin mul mul  b t sin f cos mul mul  sub
%   }
\caption{Ring $\mathcal R$.}\label{fig:ring1}
\end{figure}
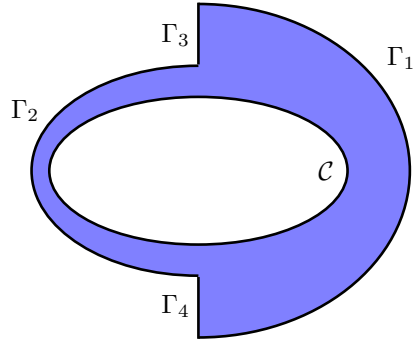

Let us examine the leaf of the phase space composed by these
trajectories.
This leaf is naturally decomposed into four rings
equal to $\mathcal R$, which are glued with each other along the
border segments. Let us describe this in detail:

\begin{itemize}
\item[$\mathcal R_1$]
This ring contains the points in the phase space that correspond
the billiard particle moving away from the caustic and the
clockwise direction around the ellipses center.

\item[$\mathcal R_2$]
Corresponds to the motion away from the caustic in the
counterclockwise direction.

\item[$\mathcal R_3$]
Corresponds to the motion towards the caustic in the
counterclockwise direction.

\item[$\mathcal R_4$]
Corresponds to the motion towards the caustic in the clockwise
direction.
\end{itemize}

Let us notice that the reflection off the two ellipse arcs
contained in the billiard boundary changes the direction of the
particle motion with respect to the caustic, but preserves the
direction of the motion around the foci. The same holds for
passing though tangency points with the caustic. On the other
hand, reflection on the axis changes the direction of motion
around the foci, but preserves the direction with respect to the
caustic. Thus, the four rings are connected to each other
according to the following scheme:
$$
 \xymatrix{
 &
\mathcal{R}_2 \ar@{-}[ld]_{\Gamma_3\Gamma_4} \ar@{-}[rd]^{\Gamma_1\Gamma_2\mathcal{C}}
 &
 \\
\mathcal{R}_1\ar@{-}[rd]_{\Gamma_1\Gamma_2\mathcal{C}}
 &
 &
\mathcal{R}_3
 \\
 &
\mathcal{R}_4 \ar@{-}[ru]_{\Gamma_3\Gamma_4} 
}
$$

Let us represent all the rings in Figures
\ref{fig:ring2} and \ref{fig:ring3}.

\begin{figure}[h]
\begin{pspicture}(-1,-0.5)(12,2.75)

\psset{fillstyle=solid, fillcolor=yellow!30, linewidth=0pt, linecolor=white}
\pspolygon(0,0)(2,0)(2,2)(0,2)%
\pspolygon(3,0)(5,0)(5,2)(3,2)%
\pspolygon(6,0)(8,0)(8,2)(6,2)%
\pspolygon(9,0)(11,0)(11,2)(9,2)

\psset{linewidth=1.3pt, ArrowInside=->, arrowscale=1.4}%
\psset{fillcolor=white, linestyle=solid}%

\psset{linecolor=red}%
\psline(2,0)(2,2)%
\psline[linestyle=dashed](5,0)(5,2)%
\psline[linestyle=dashed](8,0)(8,2)%
\psline(11,0)(11,2)%

\psset{linecolor=blue}
\psline(0,2)(0,0)%
\psline[linestyle=dashed](3,2)(3,0)%
\psline[linestyle=dashed](6,2)(6,0)%
\psline(9,2)(9,0)%

\psset{linecolor=green}
\psline(0,0)(2,0)%
\psline(3,0)(5,0)
\psline[linestyle=dashed](6,0)(8,0)%
\psline[linestyle=dashed](9,0)(11,0)

\psset{linecolor=Magenta}
\psline(2,2)(0,2)%
\psline(5,2)(3,2)%
\psline[linestyle=dashed](8,2)(6,2)%
\psline[linestyle=dashed](11,2)(9,2)%

\psset{linecolor=Orange}
\pscircle(1,1){0.5}%
\pscircle(10,1){0.5}

\pscircle[linestyle=dashed](4,1){0.5}%
\pscircle[linestyle=dashed](7,1){0.5}

\psline[linestyle=none,ArrowInside=->,ArrowInsidePos=1]
         (3,0.53)(4.15,0.5)
\psline[linestyle=none,ArrowInside=->,ArrowInsidePos=1]
         (6,0.53)(7.15,0.5)

\psline[linestyle=none,ArrowInside=->,ArrowInsidePos=1]
         (0,0.53)(1.15,0.5)
\psline[linestyle=none,ArrowInside=->,ArrowInsidePos=1]
         (9,0.53)(10.15,0.5)

\rput(2.30,1.5){$\Gamma_1$}%2.22
\rput(5.31,1.5){$\Gamma_1'$}%
\rput(8.31,1.5){$\Gamma_1'$}%
\rput(11.30,1.5){$\Gamma_1$}%

\rput(1.6,2.25){$\Gamma_3$}%
\rput(4.6,2.25){$\Gamma_3$}%
\rput(7.6,2.25){$\Gamma_3'$}%
\rput(10.6,2.25){$\Gamma_3'$}%

\rput(-0.3,0.5){$\Gamma_2$}%
\rput(2.7,0.5){$\Gamma_2'$}%
\rput(5.7,0.5){$\Gamma_2'$}%
\rput(8.7,0.5){$\Gamma_2$}%

\rput(0.4,-0.3){$\Gamma_4$}%
\rput(3.4,-0.3){$\Gamma_4$}%
\rput(6.4,-0.3){$\Gamma_4'$}%
\rput(9.4,-0.3){$\Gamma_4'$}%

\rput(1.25,1){$\mathcal{C}$}%
\rput(4.25,1){$\mathcal{C}'$}%
\rput(7.25,1){$\mathcal{C}'$}%
\rput(10.25,1){$\mathcal{C}$}%

\rput(0,2.5){$\mathcal R_1:$}%
\rput(3,2.5){$\mathcal R_2:$}%
\rput(6,2.5){$\mathcal R_3:$}%
\rput(9,2.5){$\mathcal R_4:$}%
\end{pspicture}

%\parametricplot[linecolor=red, linewidth=1pt, plotpoints=10000]{0}{360}
%   {
%    /f 30 def
%    /a 6.5 sqrt def
%    /b 2.3 sqrt def
%    a t cos f cos mul mul  b t sin f sin mul mul  add
%    a t cos f sin mul mul  b t sin f cos mul mul  sub
%   }

%arrowsize=2pt3,arrowlength=2,arrowinset=0.4,
\caption{Rings $\mathcal R_1$, $\mathcal R_2$, $\mathcal R_3$,
$\mathcal R_4$.}\label{fig:ring2}
\end{figure}

\begin{figure}[h]
\begin{pspicture}(-1,-0.1)(5,4.1)

\psset{fillstyle=solid, fillcolor=yellow!30, linewidth=0pt, linecolor=yellow!30}
\pspolygon(0,0)(4,0)(4,4)(0,4)%

\psset{linewidth=1.3pt, ArrowInside=->, arrowscale=1.4}%
\psset{fillcolor=white, linestyle=solid}%

\psset{linecolor=red}%
\psline(0,2)(0,4)%
\psline[linestyle=dashed](0,2)(0,0)%
\psline[linestyle=dashed](4,2)(4,0)%
\psline(4,2)(4,4)%

\psset{linecolor=blue}
\psline(2,4)(2,2)%
\psline[linestyle=dashed](2,0)(2,2)%

\psset{linecolor=green}
\psline(2,2)(4,2)%
\psline[linestyle=dashed](2,2)(0,2)%

\psset{linecolor=Magenta}
\psline(4,0)(2,0)%
\psline(4,4)(2,4)%
\psline[linestyle=dashed](0,0)(2,0)%
\psline[linestyle=dashed](0,4)(2,4)%

\psset{linecolor=Orange}
\pscircle(1,3){0.5}%
\pscircle(3,3){0.5}%
\pscircle[linestyle=dashed](1,1){0.5}%
\pscircle[linestyle=dashed](3,1){0.5}

\psline[linestyle=none,ArrowInside=->,ArrowInsidePos=1]
         (0,0.53)(1.15,0.5)
\psline[linestyle=none,ArrowInside=->,ArrowInsidePos=1]
         (2,2.53)(3.15,2.5)

\psline[linestyle=none,ArrowInside=->,ArrowInsidePos=1]
         (2,2.47)(0.85,2.5)
\psline[linestyle=none,ArrowInside=->,ArrowInsidePos=1]
         (4,0.47)(2.85,0.5)

\end{pspicture}

%\parametricplot[linecolor=red, linewidth=1pt, plotpoints=10000]{0}{360}
%   {
%    /f 30 def
%    /a 6.5 sqrt def
%    /b 2.3 sqrt def
%    a t cos f cos mul mul  b t sin f sin mul mul  add
%    a t cos f sin mul mul  b t sin f cos mul mul  sub
%   }

%arrowsize=2pt3,arrowlength=2,arrowinset=0.4,
\caption{Gluing rings $\mathcal R_1$, $\mathcal R_2$, $\mathcal
R_3$, $\mathcal R_4$.}\label{fig:ring3}
\end{figure}
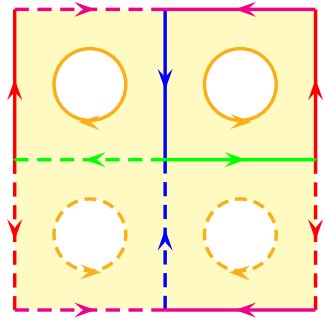

Now, we have the following

\begin{proposition}\label{prop:genus3}
All billiard trajectories within domain $\mathcal{D}_0$ with a fixed elliptical caustic form an orientable surface of genus $3$.
\end{proposition}

\section{Examples}\label{sec:examples}
In this section, we are going to analyze examples of billiards in domains described in Section \ref{sec:uvod-nosonja}, in a special case when confocal family (\ref{eq:confocal}) is degenerate: $a=b$.
The family then consists of concentric circles.
It is convenient to analyze such cases, since they may be approached with elementary means, yet all phenomena appearing in non-degenerate cases are to be preserved, as will be shown in Section \ref{sec:interval}.

\subsection{Domain bounded by circles with rotation numbers $\dfrac13$ and $\dfrac14$}
Let us consider example of the billiard within the domain with two concentric half-circles of radii $2R$, $R\sqrt{2}$ and the corresponding segments.
We will consider caustic with radius $R$.

In this case, there exist six trajectories connecting singular points corresponding to reflex angles of the boundary.
Those trajectories are represented in Figure \ref{fig:c34-sadd}.
Each polygonal line shown on the figure corresponds to two trajectories in the phase space, depending on direction of the motion. 
\begin{figure}[h]
\begin{pspicture}(-5.5,-2.1)(5.5,2.1)

  \SpecialCoor

% prva trajektorija

\psline[linewidth=1pt, linecolor=blue]
      (!0 2 sqrt)(1.93185, -0.517638) (-0.707107, -1.22474)
	(-1.22474, 0.707107)(1.41421, 1.41421)(0.517638, -1.93185)(0, -1.41421)

\psellipticarc[linecolor=black, linewidth=1pt]
    (0,0)(2,2){-90}{90}

\psellipticarc[linecolor=black, linewidth=1pt]
    (0,0)(!2 sqrt 2 sqrt){90}{270}

\psellipse[linecolor=black, linewidth=0.8pt, linestyle=dashed]%
    (0,0)(1,1)

\psline[linewidth=1pt](!0 2 sqrt)(0,2)
\psline[linewidth=1pt](!0 2 sqrt neg)(0,-2)

\pscircle*[linecolor=red](!0 2 sqrt){0.06}
\pscircle*[linecolor=red](!0 2 sqrt neg){0.06}

\pscircle*[linecolor=blue](1.93185, -0.517638){0.06}
\pscircle*[linecolor=blue](-0.707107, -1.22474){0.06}
\pscircle*[linecolor=blue](-1.22474, 0.707107){0.06}
\pscircle*[linecolor=blue](1.41421, 1.41421){0.06}
\pscircle*[linecolor=blue](0.517638, -1.93185){0.06}

% druga trajektorija

\psline[linewidth=1pt, linecolor=green]
      (!-4 2 sqrt neg)(!1.93185 4 sub 0.517638) (-4.707107, 1.22474)
	(-5.22474, -0.707107)(!1.41421 4 sub -1.41421)(!0.517638 4 sub 1.93185)	 	(-4, 1.41421)

\psellipticarc[linecolor=black, linewidth=1pt]
    (-4,0)(2,2){-90}{90}

\psellipticarc[linecolor=black, linewidth=1pt]
    (-4,0)(!2 sqrt 2 sqrt){90}{270}

\psellipse[linecolor=black, linewidth=0.8pt, linestyle=dashed]%
    (-4,0)(1,1)

\psline[linewidth=1pt](!-4 2 sqrt)(-4,2)
\psline[linewidth=1pt](!-4 2 sqrt neg)(-4,-2)

\pscircle*[linecolor=red](!-4 2 sqrt){0.06}
\pscircle*[linecolor=red](!-4 2 sqrt neg){0.06}

\pscircle*[linecolor=green](!1.93185 4 sub 0.517638){0.06}
\pscircle*[linecolor=green](-4.707107, 1.22474){0.06}
\pscircle*[linecolor=green](-5.22474,-0.707107){0.06}
\pscircle*[linecolor=green](!1.41421 4 sub -1.41421){0.06}
\pscircle*[linecolor=green](!0.517638 4 sub 1.93185){0.06}

% treca trajektorija

\psline[linewidth=1pt, linecolor=magenta]
      (!4 2 sqrt)(!2 sqrt neg 4 add 0)(4, -1.41421)

\psellipticarc[linecolor=black, linewidth=1pt]
    (4,0)(2,2){-90}{90}

\psellipticarc[linecolor=black, linewidth=1pt]
    (4,0)(!2 sqrt 2 sqrt){90}{270}

\psellipse[linecolor=black, linewidth=0.8pt, linestyle=dashed]%
    (4,0)(1,1)

\psline[linewidth=1pt](!4 2 sqrt)(4,2)
\psline[linewidth=1pt](!4 2 sqrt neg)(4,-2)

\pscircle*[linecolor=red](!4 2 sqrt){0.06}
\pscircle*[linecolor=red](!4 2 sqrt neg){0.06}

\pscircle*[linecolor=magenta](!2 sqrt neg 4 add 0){0.06}

\rput(-5.5,-1.8){a)}
\rput(-1.5,-1.8){b)}
\rput(2.5,-1.8){c)}

\end{pspicture}
\caption{Saddle-connections corresponding to circles with rotation numbers $\dfrac13$ and $\dfrac14$.}\label{fig:c34-sadd}
\end{figure}

Vertices of the saddle-connections divide the billiard border into thirteen parts, see Figure \ref{fig:c34-parts}.
For the definition of saddle-connections -- see Section \ref{sec:intro}: Introduction.
\begin{figure}[h]
\begin{pspicture}(-5.5,-2.1)(5.5,2.1)

  \SpecialCoor

\psset{linewidth=2pt}

\psellipticarc[linecolor=gray]
    (0,0)(2,2){-90}{-75}
\psellipticarc[linecolor=orange]
    (0,0)(2,2){-75}{-45}
\psellipticarc[linecolor=gray]
    (0,0)(2,2){-45}{-15}
\psellipticarc[linecolor=orange]
    (0,0)(2,2){-15}{15}
\psellipticarc[linecolor=gray]
    (0,0)(2,2){15}{45}
\psellipticarc[linecolor=orange]
    (0,0)(2,2){45}{75}
\psellipticarc[linecolor=gray]
    (0,0)(2,2){75}{90}

\psline[linecolor=gray](!0 2 sqrt)(0,2)
\psline[linecolor=gray](!0 2 sqrt neg)(0,-2)

\psellipticarc[linecolor=orange]
    (0,0)(!2 sqrt 2 sqrt){90}{120}
\psellipticarc[linecolor=gray]
    (0,0)(!2 sqrt 2 sqrt){120}{150}
\psellipticarc[linecolor=orange]
    (0,0)(!2 sqrt 2 sqrt){150}{210}
\psellipticarc[linecolor=gray]
    (0,0)(!2 sqrt 2 sqrt){210}{240}
\psellipticarc[linecolor=orange]
    (0,0)(!2 sqrt 2 sqrt){240}{270}

\pscircle*[linecolor=red](!0 2 sqrt){0.06}
\pscircle*[linecolor=red](!0 2 sqrt neg){0.06}

\pscircle*[linecolor=blue](1.93185, -0.517638){0.06}
\pscircle*[linecolor=blue](-0.707107, -1.22474){0.06}
\pscircle*[linecolor=blue](-1.22474, 0.707107){0.06}
\pscircle*[linecolor=blue](1.41421, 1.41421){0.06}
\pscircle*[linecolor=blue](0.517638, -1.93185){0.06}

\pscircle*[linecolor=green](1.93185, 0.517638){0.06}
\pscircle*[linecolor=green](-.707107, 1.22474){0.06}
\pscircle*[linecolor=green](-1.22474,-0.707107){0.06}
\pscircle*[linecolor=green](1.41421, -1.41421){0.06}
\pscircle*[linecolor=green](0.517638, 1.93185){0.06}

\pscircle*[linecolor=magenta](!2 sqrt neg 0){0.06}

\end{pspicture}
\caption{Parts of the boundary corresponding to circles with rotation numbers $\dfrac13$ and $\dfrac14$.}\label{fig:c34-parts}
\end{figure}

All trajectories in this billiard domain corresponding to the fixed caustic are periodic:
\begin{itemize}
\item
either all bouncing points of a given trajectory are in gray parts -- in this case the billiard particle hits twice each gray part until the trajectory becomes closed. Such a trajectory is $12$-periodic, with with four bounces on the smaller circle, six bounces on the bigger one, and one on each of the segments on the $y$-axis (see Figures \ref{fig:c34-periodic}a and \ref{fig:c34-periodic}c);

\item 
or all bouncing points are in orange parts -- the particle will hit each part once until closure and the trajectory is $7$-periodic. Such a trajectory hits three times the bigger circle and four times the smaller one (see Figure \ref{fig:c34-periodic}b).
\end{itemize}
\begin{figure}[h]
\begin{pspicture}(-5.5,-2.1)(5.5,2.1)

  \SpecialCoor

% prva trajektorija

\psline[linewidth=1pt, linecolor=orange]
     (-0.141186, 1.40715)(1.97388, -0.322189)(-0.581304, -1.28922)
(-1.28922, 0.581304)(1.26596, 1.54833)(0.707915, -1.87052)(-1.40715, -0.141186)
(-0.141186, 1.40715)(1.97388, -0.322189)

\psellipticarc[linecolor=black, linewidth=1pt]
    (0,0)(2,2){-90}{90}

\psellipticarc[linecolor=black, linewidth=1pt]
    (0,0)(!2 sqrt 2 sqrt){90}{270}

\psellipse[linecolor=black, linewidth=0.8pt, linestyle=dashed]%
    (0,0)(1,1)

\psline[linewidth=1pt](!0 2 sqrt)(0,2)
\psline[linewidth=1pt](!0 2 sqrt neg)(0,-2)

\pscircle*[linecolor=red](!0 2 sqrt){0.06}
\pscircle*[linecolor=red](!0 2 sqrt neg){0.06}

\pscircle*[linecolor=blue](1.93185, -0.517638){0.06}
\pscircle*[linecolor=blue](-0.707107, -1.22474){0.06}
\pscircle*[linecolor=blue](-1.22474, 0.707107){0.06}
\pscircle*[linecolor=blue](1.41421, 1.41421){0.06}
\pscircle*[linecolor=blue](0.517638, -1.93185){0.06}

\pscircle*[linecolor=green](1.93185, 0.517638){0.06}
\pscircle*[linecolor=green](-.707107, 1.22474){0.06}
\pscircle*[linecolor=green](-1.22474,-0.707107){0.06}
\pscircle*[linecolor=green](1.41421, -1.41421){0.06}
\pscircle*[linecolor=green](0.517638, 1.93185){0.06}

\pscircle*[linecolor=magenta](!2 sqrt neg 0){0.06}

% druga trajektorija

\psline[linewidth=1pt, linecolor=gray]
(-4, -1.58995)(!0.31193 4 sub -1.97553)(!1.55489 4 sub 1.2579)
(-5.14373,0.831798)
(-4.831798, -1.14373)(!1.86682 4 sub -0.717624)(-4, 1.58995)
(!0.31193 4 sub 1.97553)
(!1.55489 4 sub -1.2579)(-5.14373, -0.831798)(-4.831798, 1.14373)
(!1.86682 4 sub 0.717624)(-4, -1.58995)

\psellipticarc[linecolor=black, linewidth=1pt]
    (-4,0)(2,2){-90}{90}

\psellipticarc[linecolor=black, linewidth=1pt]
    (-4,0)(!2 sqrt 2 sqrt){90}{270}

\psellipse[linecolor=black, linewidth=0.8pt, linestyle=dashed]%
    (-4,0)(1,1)

\psline[linewidth=1pt](!-4 2 sqrt)(-4,2)
\psline[linewidth=1pt](!-4 2 sqrt neg)(-4,-2)

\pscircle*[linecolor=red](!-4 2 sqrt){0.06}
\pscircle*[linecolor=red](!-4 2 sqrt neg){0.06}

\pscircle*[linecolor=blue](!1.93185 4 sub -0.517638){0.06}
\pscircle*[linecolor=blue](-4.707107, -1.22474){0.06}
\pscircle*[linecolor=blue](-5.22474, 0.707107){0.06}
\pscircle*[linecolor=blue](!1.41421 4 sub 1.41421){0.06}
\pscircle*[linecolor=blue](!0.517638 4 sub -1.93185){0.06}

\pscircle*[linecolor=green](!1.93185 4 sub 0.517638){0.06}
\pscircle*[linecolor=green](-4.707107, 1.22474){0.06}
\pscircle*[linecolor=green](-5.22474,-0.707107){0.06}
\pscircle*[linecolor=green](!1.41421 4 sub -1.41421){0.06}
\pscircle*[linecolor=green](!0.517638 4 sub 1.93185){0.06}

\pscircle*[linecolor=magenta](!2 sqrt neg 4 sub 0){0.06}

% treca trajektorija

\psline[linewidth=1pt, linecolor=gray]
(4,-2)(5.73205, 1)(3, 1)(3, -1)(5.73205,-1)(4, 2)

\psellipticarc[linecolor=black, linewidth=1pt]
    (4,0)(2,2){-90}{90}

\psellipticarc[linecolor=black, linewidth=1pt]
    (4,0)(!2 sqrt 2 sqrt){90}{270}

\psellipse[linecolor=black, linewidth=0.8pt, linestyle=dashed]%
    (4,0)(1,1)

\psline[linewidth=1pt](!4 2 sqrt)(4,2)
\psline[linewidth=1pt](!4 2 sqrt neg)(4,-2)

\pscircle*[linecolor=red](!4 2 sqrt){0.06}
\pscircle*[linecolor=red](!4 2 sqrt neg){0.06}

\pscircle*[linecolor=blue](5.93185, -0.517638){0.06}
\pscircle*[linecolor=blue](!-0.707107 4 add -1.22474){0.06}
\pscircle*[linecolor=blue](!-1.22474 4 add 0.707107){0.06}
\pscircle*[linecolor=blue](5.41421, 1.41421){0.06}
\pscircle*[linecolor=blue](4.517638, -1.93185){0.06}

\pscircle*[linecolor=green](5.93185, 0.517638){0.06}
\pscircle*[linecolor=green](!-.707107 4 add 1.22474){0.06}
\pscircle*[linecolor=green](!-1.22474 4 add -0.707107){0.06}
\pscircle*[linecolor=green](5.41421, -1.41421){0.06}
\pscircle*[linecolor=green](4.517638, 1.93185){0.06}

\pscircle*[linecolor=magenta](!2 sqrt neg 4 add 0){0.06}

\rput(-5.5,-1.8){a)}
\rput(-1.5,-1.8){b)}
\rput(2.5,-1.8){c)}

\end{pspicture}
\caption{Periodic trajectories corresponding to circles with rotation numbers $\dfrac13$ and $\dfrac14$.}\label{fig:c34-periodic}
\end{figure}

The corresponding level set in the phase space is divided by the saddle-connections into three parts:
\begin{itemize}
\item
the part containing all $12$-periodic trajectories: this part is bounded by four saddle-connections whose projections on the configuration space is shown in Figures \ref{fig:c34-sadd}a and \ref{fig:c34-sadd}b;

\item
two parts containing all $7$-periodic trajectories winding about the caustic in the clockwise and counterclockwise direction: these parts are bounded by saddle-connections winding in the same direction whose projections on the configuration space is shown in Figures \ref{fig:c34-sadd}a and \ref{fig:c34-sadd}b; the saddle-connections corresponding to Figure \ref{fig:c34-sadd}c are lying within the corresponding parts.
\end{itemize}

\subsection{Domain bounded by circles with rotation numbers $\dfrac14$ and $\dfrac16$}
Let us consider example of the billiard within the domain determined with two concentric half-circles with rotation numbers equal to $\dfrac14$ and $\dfrac16$.

In this case, there exist six saddle-connections, represented in Figure 
\ref{fig:c46-sadd}.
Each polygonal line shown on the figure corresponds to two trajectories in the phase space, depending on direction of the motion. 
\begin{figure}[h]
\begin{pspicture}(-5.5,-2.1)(5.5,2.1)

  \SpecialCoor

% prva trajektorija

\psline[linewidth=1pt, linecolor=blue]
      (!0 -1.1547 2 sqrt mul)(!0.366025 2 sqrt mul -1.36603 2 sqrt mul)
	(!1.36603 2 sqrt mul 0.366025 2 sqrt mul)(!0 1.1547 2 sqrt mul)

\psellipticarc[linecolor=black, linewidth=1pt]
    (0,0)(!2 3 div 6 sqrt mul 2 3 div 6 sqrt mul){90}{270}

\psellipticarc[linecolor=black, linewidth=1pt]    (0,0)(2,2){-90}{90}

\psellipse[linecolor=black, linewidth=0.8pt, linestyle=dashed]%
    (0,0)(!2 sqrt 2 sqrt)

\psline[linewidth=1pt](!0 2 3 div 6 sqrt mul)(0,2)
\psline[linewidth=1pt](!0 2 3 div 6 sqrt mul neg)(0,-2)

\pscircle*[linecolor=red](!0 2 3 div 6 sqrt mul){0.06}
\pscircle*[linecolor=red](!0 2 3 div 6 sqrt mul neg){0.06}

\pscircle*[linecolor=blue](!0.366025 2 sqrt mul -1.36603 2 sqrt mul){0.06}
\pscircle*[linecolor=blue](!1.36603 2 sqrt mul 0.366025 2 sqrt mul){0.06}

% druga trajektorija

\psline[linewidth=1pt, linecolor=green]
      (!-4 1.1547 2 sqrt mul)(!0.366025 2 sqrt mul 4 sub 1.36603 2 sqrt mul)
	(!1.36603 2 sqrt mul 4 sub -0.366025 2 sqrt mul)(!-4 -1.1547 2 sqrt mul)

\psellipticarc[linecolor=black, linewidth=1pt]
    (-4,0)(!2 3 div 6 sqrt mul 2 3 div 6 sqrt mul){90}{270}

\psellipticarc[linecolor=black, linewidth=1pt]    (-4,0)(2,2){-90}{90}

\psellipse[linecolor=black, linewidth=0.8pt, linestyle=dashed]%
    (-4,0)(!2 sqrt 2 sqrt)

\psline[linewidth=1pt](!-4 2 3 div 6 sqrt mul)(-4,2)
\psline[linewidth=1pt](!-4 2 3 div 6 sqrt mul neg)(-4,-2)

\pscircle*[linecolor=red](!-4 2 3 div 6 sqrt mul){0.06}
\pscircle*[linecolor=red](!-4 2 3 div 6 sqrt mul neg){0.06}

\pscircle*[linecolor=green](!0.366025 2 sqrt mul 4 sub 1.36603 2 sqrt mul){0.06}
\pscircle*[linecolor=green](!1.36603 2 sqrt mul 4 sub -0.366025 2 sqrt mul){0.06}

% treca trajektorija

\psline[linewidth=1pt, linecolor=magenta]
      (!4 1.1547 2 sqrt mul)(!2 sqrt neg 4 add 0.57735 2 sqrt mul)
     (!2 sqrt neg 4 add -0.57735 2 sqrt mul)(!4 -1.1547 2 sqrt mul)

\psellipticarc[linecolor=black, linewidth=1pt]
    (4,0)(!2 3 div 6 sqrt mul 2 3 div 6 sqrt mul){90}{270}

\psellipticarc[linecolor=black, linewidth=1pt]    (4,0)(2,2){-90}{90}

\psellipse[linecolor=black, linewidth=0.8pt, linestyle=dashed]%
    (4,0)(!2 sqrt 2 sqrt)

\psline[linewidth=1pt](!4 2 3 div 6 sqrt mul)(4,2)
\psline[linewidth=1pt](!4 2 3 div 6 sqrt mul neg)(4,-2)

\pscircle*[linecolor=red](!4 2 3 div 6 sqrt mul){0.06}
\pscircle*[linecolor=red](!4 2 3 div 6 sqrt mul neg){0.06}

\pscircle*[linecolor=magenta](!2 sqrt neg 4 add 0.57735 2 sqrt mul){0.06}
\pscircle*[linecolor=magenta](!2 sqrt neg 4 add -0.57735 2 sqrt mul){0.06}

\rput(-5.5,-1.8){a)}
\rput(-1.5,-1.8){b)}
\rput(2.5,-1.8){c)}

\end{pspicture}
\caption{Saddle-connections corresponding to circles with rotation numbers $\dfrac14$ and $\dfrac16$.}\label{fig:c46-sadd}
\end{figure}

Vertices of the saddle-connections divide the billiard border into eight parts, see Figure \ref{fig:c46-parts}.
\begin{figure}[h]
\begin{pspicture}(-5.5,-2.1)(5.5,2.1)

  \SpecialCoor

\psset{linewidth=2pt}

\psellipticarc[linecolor=orange]
    (0,0)(!2 3 div 6 sqrt mul 2 3 div 6 sqrt mul){90}{270}

\psellipticarc[linecolor=gray](0,0)(2,2){-90}{-75}

\psellipticarc[linecolor=orange](0,0)(2,2){-75}{-15}

\psellipticarc[linecolor=gray](0,0)(2,2){-15}{15}

\psellipticarc[linecolor=orange](0,0)(2,2){15}{75}

\psellipticarc[linecolor=gray](0,0)(2,2){75}{90}

\psline[linecolor=gray](!0 2 3 div 6 sqrt mul)(0,2)
\psline[linecolor=gray](!0 2 3 div 6 sqrt mul neg)(0,-2)

\pscircle*[linecolor=red](!0 2 3 div 6 sqrt mul){0.06}
\pscircle*[linecolor=red](!0 2 3 div 6 sqrt mul neg){0.06}

\pscircle*[linecolor=blue](!0.366025 2 sqrt mul -1.36603 2 sqrt mul){0.06}
\pscircle*[linecolor=blue](!1.36603 2 sqrt mul 0.366025 2 sqrt mul){0.06}

\pscircle*[linecolor=green](!0.366025 2 sqrt mul 1.36603 2 sqrt mul){0.06}
\pscircle*[linecolor=green](!1.36603 2 sqrt mul -0.366025 2 sqrt mul){0.06}

\pscircle*[linecolor=magenta](!2 sqrt neg 0.57735 2 sqrt mul){0.06}
\pscircle*[linecolor=magenta](!2 sqrt neg -0.57735 2 sqrt mul){0.06}

\end{pspicture}
\caption{Parts of the boundary corresponding to circles with rotation numbers $\dfrac14$ and $\dfrac16$.}\label{fig:c46-parts}
\end{figure}

All trajectories in this billiard domain corresponding to the fixed caustic are periodic:
\begin{itemize}
\item
either all bouncing points of a given trajectory are in gray parts -- in this case the billiard particle hits twice each gray part until the trajectory becomes closed and the trajectory is $6$-periodic, see Figures \ref{fig:c46-periodic}a and \ref{fig:c46-periodic}c;

\item 
or all bouncing points are in orange parts -- the particle will hit each part once until closure and the trajectory is $5$-periodic, see Figure \ref{fig:c46-periodic}b.
\end{itemize}
\begin{figure}[h]
\begin{pspicture}(-5.5,-2.1)(5.5,2.1)

  \SpecialCoor

% prva trajektorija

\psline[linecolor=orange, linewidth=1pt]
(!-0.341237 2 sqrt mul 1.10313 2 sqrt mul)(!-1.12596 2 sqrt mul 0.256044 2 sqrt mul)
(!-0.784718 2 sqrt mul -0.847084 2 sqrt mul)(!0.753366 2 sqrt mul -1.19685 2 sqrt mul)
(!1.19685 2 sqrt mul 0.753366 2 sqrt mul)(!-0.341237 2 sqrt mul 1.10313 2 sqrt mul)

\psellipticarc[linecolor=black, linewidth=1pt]
    (0,0)(!2 3 div 6 sqrt mul 2 3 div 6 sqrt mul){90}{270}

\psellipticarc[linecolor=black, linewidth=1pt]    (0,0)(2,2){-90}{90}

\psellipse[linecolor=black, linewidth=0.8pt, linestyle=dashed]%
    (0,0)(!2 sqrt 2 sqrt)

\psline[linewidth=1pt](!0 2 3 div 6 sqrt mul)(0,2)
\psline[linewidth=1pt](!0 2 3 div 6 sqrt mul neg)(0,-2)

\pscircle*[linecolor=red](!0 2 3 div 6 sqrt mul){0.06}
\pscircle*[linecolor=red](!0 2 3 div 6 sqrt mul neg){0.06}

\pscircle*[linecolor=blue](!0.366025 2 sqrt mul -1.36603 2 sqrt mul){0.06}
\pscircle*[linecolor=blue](!1.36603 2 sqrt mul 0.366025 2 sqrt mul){0.06}

\pscircle*[linecolor=green](!0.366025 2 sqrt mul 1.36603 2 sqrt mul){0.06}
\pscircle*[linecolor=green](!1.36603 2 sqrt mul -0.366025 2 sqrt mul){0.06}

\pscircle*[linecolor=magenta](!2 sqrt neg 0.57735 2 sqrt mul){0.06}
\pscircle*[linecolor=magenta](!2 sqrt neg -0.57735 2 sqrt mul){0.06}

% druga trajektorija

\psline[linecolor=gray, linewidth=1pt]
(!1.40715 2 sqrt mul 4 sub 0.141186 2 sqrt mul)(!-4 1.29171 2 sqrt mul)
(!0.141186 2 sqrt mul 4 sub 1.40715 2 sqrt mul)
(!1.40715 2 sqrt mul 4 sub -0.141186 2 sqrt mul)
(!-4 -1.29171 2 sqrt mul)(!0.141186 2 sqrt mul 4 sub -1.40715 2 sqrt mul)
(!1.40715 2 sqrt mul 4 sub 0.141186 2 sqrt mul)

\psellipticarc[linecolor=black, linewidth=1pt]
    (-4,0)(!2 3 div 6 sqrt mul 2 3 div 6 sqrt mul){90}{270}

\psellipticarc[linecolor=black, linewidth=1pt]    (-4,0)(2,2){-90}{90}

\psellipse[linecolor=black, linewidth=0.8pt, linestyle=dashed]%
    (-4,0)(!2 sqrt 2 sqrt)

\psline[linewidth=1pt](!-4 2 3 div 6 sqrt mul)(-4,2)
\psline[linewidth=1pt](!-4 2 3 div 6 sqrt mul neg)(-4,-2)

\pscircle*[linecolor=red](!-4 2 3 div 6 sqrt mul){0.06}
\pscircle*[linecolor=red](!-4 2 3 div 6 sqrt mul neg){0.06}

\pscircle*[linecolor=green](!0.366025 2 sqrt mul 4 sub 1.36603 2 sqrt mul){0.06}
\pscircle*[linecolor=green](!1.36603 2 sqrt mul 4 sub -0.366025 2 sqrt mul){0.06}

\pscircle*[linecolor=blue](!0.366025 2 sqrt mul 4 sub -1.36603 2 sqrt mul){0.06}
\pscircle*[linecolor=blue](!1.36603 2 sqrt mul 4 sub 0.366025 2 sqrt mul){0.06}

\pscircle*[linecolor=magenta](!2 sqrt neg 4 sub 0.57735 2 sqrt mul){0.06}
\pscircle*[linecolor=magenta](!2 sqrt neg 4 sub -0.57735 2 sqrt mul){0.06}

% treca trajektorija
\psline[linecolor=gray, linewidth=1pt]
(4,-2)(6,0)(4,2)

\psellipticarc[linecolor=black, linewidth=1pt]
    (4,0)(!2 3 div 6 sqrt mul 2 3 div 6 sqrt mul){90}{270}

\psellipticarc[linecolor=black, linewidth=1pt]    (4,0)(2,2){-90}{90}

\psellipse[linecolor=black, linewidth=0.8pt, linestyle=dashed]%
    (4,0)(!2 sqrt 2 sqrt)

\psline[linewidth=1pt](!4 2 3 div 6 sqrt mul)(4,2)
\psline[linewidth=1pt](!4 2 3 div 6 sqrt mul neg)(4,-2)

\pscircle*[linecolor=red](!4 2 3 div 6 sqrt mul){0.06}
\pscircle*[linecolor=red](!4 2 3 div 6 sqrt mul neg){0.06}

\pscircle*[linecolor=magenta](!2 sqrt neg 4 add 0.57735 2 sqrt mul){0.06}
\pscircle*[linecolor=magenta](!2 sqrt neg 4 add -0.57735 2 sqrt mul){0.06}

\pscircle*[linecolor=blue](!0.366025 2 sqrt mul 4 add -1.36603 2 sqrt mul){0.06}
\pscircle*[linecolor=blue](!1.36603 2 sqrt mul 4 add 0.366025 2 sqrt mul){0.06}

\pscircle*[linecolor=green](!0.366025 2 sqrt mul 4 add 1.36603 2 sqrt mul){0.06}
\pscircle*[linecolor=green](!1.36603 2 sqrt mul 4 add -0.366025 2 sqrt mul){0.06}

\rput(-5.5,-1.8){a)}
\rput(-1.5,-1.8){b)}
\rput(2.5,-1.8){c)}

\end{pspicture}
\caption{Periodic trajectories corresponding to circles with rotation numbers $\dfrac14$ and $\dfrac16$.}\label{fig:c46-periodic}
\end{figure}

The corresponding level set in the phase space is divided by the saddle-connections into three parts:
\begin{itemize}
\item
the part containing all $6$-periodic trajectories: this part is bounded by four saddle-connections whose projections on the configuration space is shown on 
Figures \ref{fig:c46-sadd}a and \ref{fig:c46-sadd}b;

\item
two parts containing all $5$-periodic trajectories winding about the caustic in the clockwise and counterclockwise direction: these parts are bounded by saddle-connections winding in the same direction whose projections on the configuration space is shown in Figures \ref{fig:c46-sadd}a and \ref{fig:c46-sadd}b; the saddle-connections corresponding to Figure \ref{fig:c46-sadd}c are lying within the corresponding parts.
\end{itemize}

%\subsection{Domain bounded by circles circumscribed about a triangle and a pentagon}
%\input{4-examples/circles35}

\subsection{Domain bounded by circles with rotation numbers $\dfrac{5-\sqrt5}{10}$ and $\dfrac{\sqrt5}{10}$}
In this example, there exist six saddle-connections, represented in Figure 
\ref{fig:c5-sadd}.
Each polygonal line shown on the figure corresponds to two trajectories in the phase space, depending on direction of the motion. 
\begin{figure}[h]
\begin{pspicture}(-5.5,-2.1)(5.5,2.1)

  \SpecialCoor

% prva trajektorija

\psline[linewidth=1pt, linecolor=blue]
      (0, 1.69308)(2,0) (0, -1.69308)

\psellipticarc[linecolor=black, linewidth=1pt]
    (0,0)(2,2){-90}{90}

\psellipticarc[linecolor=black, linewidth=1pt]
    (0,0)(1.69308,1.69308){90}{270}

\psellipse[linecolor=black, linewidth=0.8pt, linestyle=dashed]%
    (0,0)(1.29223,1.29223)

\psline[linewidth=1pt](0, 1.69308)(0,2)
\psline[linewidth=1pt](0, -1.69308)(0,-2)

\pscircle*[linecolor=red](0, 1.69308){0.06}
\pscircle*[linecolor=red](0, -1.69308){0.06}

\pscircle*[linecolor=blue](2,0){0.06}

% druga trajektorija

\psline[linewidth=1pt, linecolor=green]
      (-4, 1.69308)(-5.66985, 0.279484)(-4.551299, -1.60081)(-2.109, -0.651239)
(-3.66985, 1.97256)(-4, 1.69308)

\psellipticarc[linecolor=black, linewidth=1pt]
    (-4,0)(2,2){-90}{90}

\psellipticarc[linecolor=black, linewidth=1pt]
    (-4,0)(1.69308,1.69308){90}{270}

\psellipse[linecolor=black, linewidth=0.8pt, linestyle=dashed]%
    (-4,0)(1.29223,1.29223)

\psline[linewidth=1pt](-4, 1.69308)(-4,2)
\psline[linewidth=1pt](-4, -1.69308)(-4,-2)

\pscircle*[linecolor=red](-4, 1.69308){0.06}
\pscircle*[linecolor=red](-4, -1.69308){0.06}

\pscircle*[linecolor=green](-5.66985, 0.279484){0.06}
\pscircle*[linecolor=green](-4.551299, -1.60081){0.06}
\pscircle*[linecolor=green](-2.109, -0.651239){0.06}
\pscircle*[linecolor=green](-3.66985, 1.97256){0.06}

% treca trajektorija

\psline[linewidth=1pt, linecolor=magenta]
      (4, -1.69308)(2.33015, -0.279484)(3.4487, 1.60081)(5.891, 0.651239)
(4.33015, -1.97256)(4, -1.69308)

\psellipticarc[linecolor=black, linewidth=1pt]
    (4,0)(2,2){-90}{90}

\psellipticarc[linecolor=black, linewidth=1pt]
    (4,0)(1.69308,1.69308){90}{270}

\psellipse[linecolor=black, linewidth=0.8pt, linestyle=dashed]%
    (4,0)(1.29223,1.29223)

\psline[linewidth=1pt](4, 1.69308)(4,2)
\psline[linewidth=1pt](4, -1.69308)(4,-2)

\pscircle*[linecolor=red](4, 1.69308){0.06}
\pscircle*[linecolor=red](4, -1.69308){0.06}

\pscircle*[linecolor=magenta](2.33015, -0.279484){0.06}
\pscircle*[linecolor=magenta](3.4487, 1.60081){0.06}
\pscircle*[linecolor=magenta](5.891, 0.651239){0.06}
\pscircle*[linecolor=magenta](4.33015, -1.97256){0.06}

\rput(-5.5,-1.8){a)}
\rput(-1.5,-1.8){b)}
\rput(2.5,-1.8){c)}

\end{pspicture}
\caption{Saddle-connections corresponding to circles with rotation numbers $\dfrac{5-\sqrt5}{10}$ and $\dfrac{\sqrt5}{10}$.}\label{fig:c5-sadd}
\end{figure}
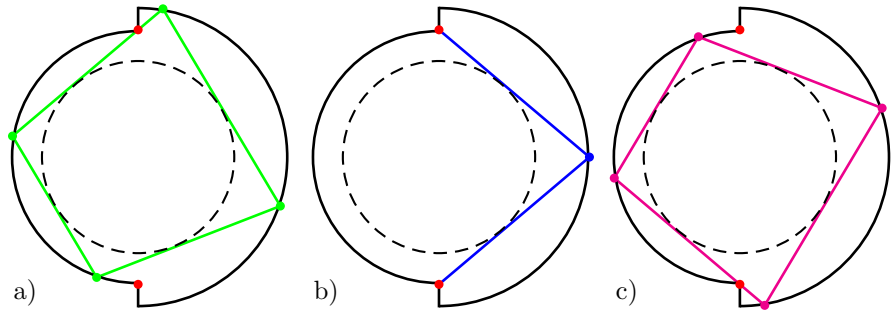

Vertices of the saddle-connections divide the billiard border into eleven parts, see Figure \ref{fig:c5-parts}.
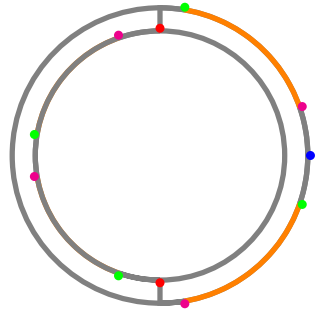
\begin{figure}[h]
\begin{pspicture}(-5.5,-2.1)(5.5,2.1)

  \SpecialCoor

\psset{linewidth=2pt}

\psellipticarc[linecolor=gray]
    (0,0)(2,2){-90}{-80.4984}
\psellipticarc[linecolor=orange]
    (0,0)(2,2){-80.4984}{-19.0031}
\psellipticarc[linecolor=gray]
    (0,0)(2,2){-19.0031}{19.0031}
\psellipticarc[linecolor=orange]
    (0,0)(2,2){19.0031}{80.4984}
\psellipticarc[linecolor=gray]
    (0,0)(2,2){80.4984}{90}

\psellipticarc[linecolor=gray]
    (0,0)(1.69308,1.69308){90}{109.003}
\psellipticarc[linecolor=orange]
    (0,0)(1.69308,1.69308){109.003}{170.498}
\psellipticarc[linecolor=gray]
    (0,0)(1.69308,1.69308){170.498}{189.502}
\psellipticarc[linecolor=orange]
    (0,0)(1.69308,1.69308){189.502}{250.997}
\psellipticarc[linecolor=gray]
    (0,0)(1.69308,1.69308){250.997}{270}

\psline[linecolor=gray](0, 1.69308)(0,2)
\psline[linecolor=gray](0, -1.69308)(0,-2)

\pscircle*[linecolor=red](0, 1.69308){0.06}
\pscircle*[linecolor=red](0, -1.69308){0.06}

\pscircle*[linecolor=blue](2,0){0.06}

\pscircle*[linecolor=green](-1.66985, 0.279484){0.06}
\pscircle*[linecolor=green](-0.551299, -1.60081){0.06}
\pscircle*[linecolor=green](1.891, -0.651239){0.06}
\pscircle*[linecolor=green](0.330149, 1.97256){0.06}

\pscircle*[linecolor=magenta](-1.66985, -0.279484){0.06}
\pscircle*[linecolor=magenta](-0.551299, 1.60081){0.06}
\pscircle*[linecolor=magenta](1.891, 0.651239){0.06}
\pscircle*[linecolor=magenta](0.33015, -1.97256){0.06}

\end{pspicture}
\caption{Parts of the boundary corresponding to circles with rotation numbers $\dfrac{5-\sqrt5}{10}$ and $\dfrac{\sqrt5}{10}$.}\label{fig:c5-parts}
\end{figure}

All trajectories in this billiard domain corresponding to the fixed caustic are periodic:
\begin{itemize}
\item
either all bouncing points of a given trajectory are in gray parts -- in this case the billiard particle hits twice each gray part until the trajectory becomes closed and the trajectory is $14$-periodic,
see Figures \ref{fig:c5-periodic}a and \ref{fig:c5-periodic}c.
Notice that such a trajectory bounces six times on each of the circles and once on each of the segments; 

\item 
or all bouncing points are in orange parts -- the particle will hit each part once until closure and the trajectory is $4$-periodic, see Figure \ref{fig:c5-periodic}b.
Such a trajectory reflects twice on each of the circular arcs.
\end{itemize}
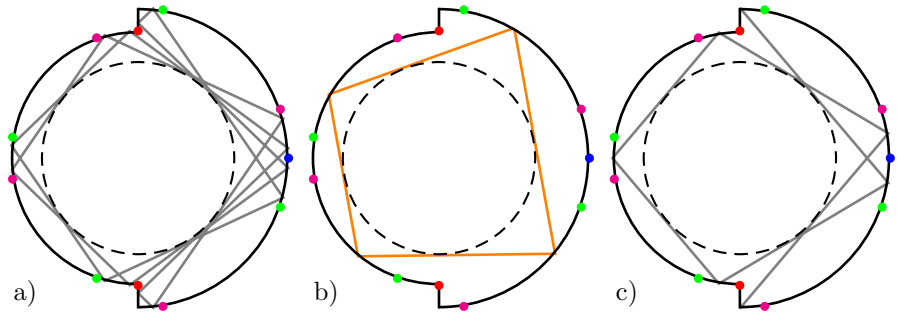
\begin{figure}[h]
\begin{pspicture}(-5.5,-2.1)(5.5,2.1)

  \SpecialCoor

% prva trajektorija

\psline[linewidth=1pt, linecolor=orange]
     (1, 1.73205)(-1.46625, 0.846539)(-1.07697, -1.30639)
(1.54321,-1.2722)(1, 1.73205)

\psellipticarc[linecolor=black, linewidth=1pt]
    (0,0)(2,2){-90}{90}

\psellipticarc[linecolor=black, linewidth=1pt]
    (0,0)(1.69308,1.69308){90}{270}

\psellipse[linecolor=black, linewidth=0.8pt, linestyle=dashed]%
    (0,0)(1.29223,1.29223)

\psline[linewidth=1pt](0, 1.69308)(0,2)
\psline[linewidth=1pt](0, -1.69308)(0,-2)

\pscircle*[linecolor=red](0, 1.69308){0.06}
\pscircle*[linecolor=red](0, -1.69308){0.06}

\pscircle*[linecolor=blue](2,0){0.06}

\pscircle*[linecolor=green](-1.66985, 0.279484){0.06}
\pscircle*[linecolor=green](-0.551299, -1.60081){0.06}
\pscircle*[linecolor=green](1.891, -0.651239){0.06}
\pscircle*[linecolor=green](0.330149, 1.97256){0.06}

\pscircle*[linecolor=magenta](-1.66985, -0.279484){0.06}
\pscircle*[linecolor=magenta](-0.551299, 1.60081){0.06}
\pscircle*[linecolor=magenta](1.891, 0.651239){0.06}
\pscircle*[linecolor=magenta](0.33015, -1.97256){0.06}

%druga trajektorija

\psline[linewidth=1pt, linecolor=gray]
(-5.68462, -0.169026)(-4.111381, -1.68941)(-2.00433, -0.131572)(-4, 1.79705)
(-3.80033, 1.99001)(-2.07025, -0.525427)(-4.444795,-1.63361)(-5.68462, 0.169026)(-4.111381, 1.68941)(-2.00433, 0.131572)(-4, -1.79705)(-3.80033, -1.99001)
(-2.07025, 0.525427)(-4.444795, 1.63361)(-5.68462, -0.169026)

\psellipticarc[linecolor=black, linewidth=1pt]
    (-4,0)(2,2){-90}{90}

\psellipticarc[linecolor=black, linewidth=1pt]
    (-4,0)(1.69308,1.69308){90}{270}

\psellipse[linecolor=black, linewidth=0.8pt, linestyle=dashed]%
    (-4,0)(1.29223,1.29223)

\psline[linewidth=1pt](-4, 1.69308)(-4,2)
\psline[linewidth=1pt](-4, -1.69308)(-4,-2)

\pscircle*[linecolor=red](-4, 1.69308){0.06}
\pscircle*[linecolor=red](-4, -1.69308){0.06}

\pscircle*[linecolor=green](-5.66985, 0.279484){0.06}
\pscircle*[linecolor=green](-4.551299, -1.60081){0.06}
\pscircle*[linecolor=green](-2.109, -0.651239){0.06}
\pscircle*[linecolor=green](-3.66985, 1.97256){0.06}

\pscircle*[linecolor=blue](-2,0){0.06}

\pscircle*[linecolor=magenta](-5.66985, -0.279484){0.06}
\pscircle*[linecolor=magenta](-4.551299, 1.60081){0.06}
\pscircle*[linecolor=magenta](-2.109, 0.651239){0.06}
\pscircle*[linecolor=magenta](-3.66985, -1.97256){0.06}

% treca trajektorija

\psline[linewidth=1pt, linecolor=gray]
(4, 2)(5.97256, -0.330149)(3.72052, -1.66985)(2.30692,0)(3.72052, 1.66985)
(5.97256, 0.330149)(4, -2)

\psellipticarc[linecolor=black, linewidth=1pt]
    (4,0)(2,2){-90}{90}

\psellipticarc[linecolor=black, linewidth=1pt]
    (4,0)(1.69308,1.69308){90}{270}

\psellipse[linecolor=black, linewidth=0.8pt, linestyle=dashed]%
    (4,0)(1.29223,1.29223)

\psline[linewidth=1pt](4, 1.69308)(4,2)
\psline[linewidth=1pt](4, -1.69308)(4,-2)

\pscircle*[linecolor=red](4, 1.69308){0.06}
\pscircle*[linecolor=red](4, -1.69308){0.06}

\pscircle*[linecolor=magenta](2.33015, -0.279484){0.06}
\pscircle*[linecolor=magenta](3.4487, 1.60081){0.06}
\pscircle*[linecolor=magenta](5.891, 0.651239){0.06}
\pscircle*[linecolor=magenta](4.33015, -1.97256){0.06}

\pscircle*[linecolor=blue](6,0){0.06}

\pscircle*[linecolor=green](2.33015, 0.279484){0.06}
\pscircle*[linecolor=green](3.4487, -1.60081){0.06}
\pscircle*[linecolor=green](5.891, -0.651239){0.06}
\pscircle*[linecolor=green](4.330149, 1.97256){0.06}

\rput(-5.5,-1.8){a)}
\rput(-1.5,-1.8){b)}
\rput(2.5,-1.8){c)}

\end{pspicture}
\caption{Periodic trajectories corresponding to circles with rotation numbers $\dfrac{5-\sqrt5}{10}$ and $\dfrac{\sqrt5}{10}$.}\label{fig:c5-periodic}
\end{figure}

The corresponding level set in the phase space is divided by the saddle-connections into three parts:
\begin{itemize}
\item
the part containing all $14$-periodic trajectories: this part is bounded by four saddle-connections whose projections on the configuration space is shown on 
Figures \ref{fig:c5-sadd}a and \ref{fig:c5-sadd}c.
The saddle-connections corresponding to Figure \ref{fig:c5-sadd}b are lying inside this part;

\item
two parts containing all $4$-periodic trajectories winding about the caustic in the clockwise and counterclockwise direction: these parts are bounded by saddle-connections winding in the same direction whose projections on the configuration space is shown in Figures \ref{fig:c5-sadd}a and \ref{fig:c5-sadd}c.
\end{itemize}

\section{General definitions and topological estimates}
Let $\mathcal{D}$ be a bounded domain in the plane such that its boundary $\Gamma=\partial\mathcal{D}$ is the union of finitely many arcs of confocal conics from the family (\ref{eq:confocal}).

We consider the billiard system within $\mathcal{D}$.
Any trajectory of this billiard will have a caustic -- a conic from (\ref{eq:confocal}) touching all lines containing segments of the trajectory.
Let us fix $\mathcal{C}_{\lambda_0}$ as caustic.

Notice that all tangent lines of a conic fill out an infinite domain in the plane: if the conic is an ellipse, the domain is its exterior; for a hyperbola, it is the part of the plane between its branches.

Denote by $\mathcal{D}_{\lambda_0}$ the intersection of $\mathcal{D}$ with the domain containing tangent lines of caustic $\mathcal{C}_{\lambda_0}$.
All billiard trajectories with caustic $\mathcal{C}_{\lambda_0}$ are placed in $\mathcal{D}_{\lambda_0}$.
$\mathcal{D}_{\lambda_0}$ is a bounded set whose boundary $\Gamma_{\lambda_0}=\partial\mathcal{D}_{\lambda_0}$
is the union of finitely many arcs of conics from (\ref{eq:confocal}).
We assume that $\mathcal{D}_{\lambda_0}$ is connected as well, otherwise we consider its connected component.

All billiard trajectories in domain $\mathcal{D}$ with the caustic $\mathcal{C}_{\lambda_0}$ will correspond to a certain compact leaf $\mathcal{M}(\lambda_0)$ in the phase space.
$\mathcal{M}_{\lambda_0}$ is obtained by gluing four copies of $\mathcal{D}_{\lambda_0}$ along the corresponding arcs of the boundary $\Gamma_{\lambda_0}=\partial\mathcal{D}_{\lambda_0}$, similarly as it is explained in Section \ref{sec:uvod-nosonja}.

On $\mathcal{M}_{\lambda_0}$, singular points of the billiard flow correspond to vertices of reflex angles on the boundary of $\mathcal{D}_{\lambda_0}$.
Since confocal conics are orthogonal to each other at points of intersection, only two types of such angles may appear: full angles and angles of $270^{\circ}$.
A vertex of a full angle is the projection of two singular points in the phase space.
Each of the singular points has four separatrices.
On the other hand, a vertex of a $270^{\circ}$ is a projection of only one singular point having six separatrices.

Using the Euler-Poincar\'e formula, as in \cite{Viana}, we get the following estimate for the total number $N=N(\mathcal{M}_{\lambda_0})$ of saddle-connections:

\begin{proposition}\label{prop:N(M)}
The total number $N=N(\mathcal{M}_{\lambda_0})$ of saddle-connections is bounded from above:
$$
N(\mathcal{M}_{\lambda_0})\le\frac{1}{2}\sum_{i=1}^ks_i=k-\chi(\mathcal{M}_{\lambda_0}),
$$
where $k$ is the number of singular points of the flow on $\mathcal{M}_{\lambda_0}$, and $s_1$, \dots, $s_k$ numbers of separatrices at each singular point.  
\end{proposition}

As a corollary, we get the following

\begin{proposition}
Consider billiard within $\mathcal{D}$ with $\mathcal{C}_{\lambda_0}$ as a caustic.
If the corresponding subdomain $\mathcal{D}_{\lambda_0}$ has $\tilde{k}$ reflex angles on its boundary $\Gamma_{\lambda_0}$ then: 
\begin{itemize}
\item $N\le 3\tilde{k}$;
\item $g(\mathcal{M}_{\lambda_0})=\tilde{k}+1$.
\end{itemize}
\end{proposition}

Notice that the genus of the surface $\mathcal{M}_{\lambda_0}$ depends only on the number of reflex angles on the boundary of $\mathcal{D}_{\lambda_0}$ and not of their types.
Also, $\tilde{k}\le k$.

%\begin{corollary}
%For billiards in domains bounded by arcs of several confocal conics, the number of saddle-connections with a fixed caustic is uniformly bounded by genus: $N\le 3g-3$.
%\end{corollary}

\begin{example}
\begin{itemize}
\item
If there are no reflex angles on the boundary, i.e.~ $k=0$, then $\mathcal{M}_{\lambda_0}$ is a torus: $g=1$, $N=0$;

\item
if there is only one reflex angle on the boundary, independently if it is a $270^{\circ}$ angle or a full angle, we have that $g=2$.
\end{itemize}
\end{example}

It is well known that the Liouville-Arnold theorem (see \cite{ArnoldMMM}) describes regular compact leaves of a completely integrable Hamiltonian system as tori, with dymanics being quasi-periodic on these inavriant tori.
On some of the tori, the dynamics is exclusively periodic and for a fixed such a torus, the period is fixed.

We finish this section by formulating of an analogue of the Liuoville-Arnold theorem for pseudo-integrable billiard systems.
It is a consequence of the Maier theorem from the theory of measured foliations (see \cite{Maier1943,Viana}).
In our case, the measured foliation is defined by the kernel of an exact one-form $\beta_{\lambda}:=d K_{\lambda}$, where functions $K_{\lambda}$ are defined in Section \ref{sec:ellbilliards}.

\begin{theorem}\label{th:maier}
There exist paiwise disjoint open domains $D_1$, \dots, $D_n$ on $\mathcal{M}_{\lambda_0}$, each of them being invariant under the billiard flow, such that their closures cover $\mathcal{M}_{\lambda_0}$ and for each $j\in\{1,\dots,n\}$:
\begin{itemize}
 \item
either $D_j$ consists of periodic billiard trajectories and is homeomorphic to a cylinder;
 \item
or $D_j$ consists of non-periodic trajectories all of which are dense in $D_j$.
\end{itemize}
The boundary of each $D_j$ consists of saddle-connections.
\end{theorem}

We see, that in contrast to completely integrable Hamiltonian systems, compact leaves of our billiards could be of a genus greater than $1$.
Moreover, one leaf could contain regions with periodic trajectories with different periods for different regions and simultaneousely could contain regions with non-periodic motion.
Because of that, we call such systems \emph{pseudo-integrable}, taking into account the fact that they possess two independent commuting first integrals, as it has been shown in Section \ref{sec:ellbilliards}.

\section{Poncelet theorem and Cayley-type conditions}\label{sec:cayley}
For billiards within confocal conics without reflex angles on the boundary, it is well known that the famous Poncelet porism holds (see \cites{DragRadn2006,DragRadn2011book}):
\begin{itemize}
\item[(A)] 
if there is a periodic billiard trajectory with one initial point of the boundary, then there are infinitely many such periodic trajectories with the same period, sharing the same caustic; 
\item[(B)]
even more is true, if there is one periodic trajectory, then all trajectories sharing the same caustic are periodic with the same period.
\end{itemize}

However, when reflex angles exist, which is the case studied in the present paper, one can say that (A) is still generally true.
However, (B) is not true any more.
In other words, \emph{the Poncelet porism is true locally, but not globally}.

\begin{theorem}\label{th:poncelet}
There exist subsets $\delta_1$, \dots, $\delta_n$ of the boundary $\Gamma_{\lambda_0}$, with the following properties:
\begin{itemize}
 \item
$\delta_1$, \dots, $\delta_n$ are invariant under the billiard map; 

 \item 
$\delta_1$, \dots, $\delta_n$ are pairwise disjoint;
 
 \item
each $\delta_i$ is a finite union of $d_i$ open subarcs of $\Gamma_{\lambda_0}$: 
$$
\delta_i=\bigcup_{j=1}^{d_i}\ell_{j}^i;
$$

 \item
closure of $\delta_1\cup\dots\cup \delta_N$ is $\Gamma$,
\end{itemize}
such that they satisfy:
\begin{itemize}
 \item
if one billiard trajectory with bouncing points within $\delta_i$ is periodic, then all such trajectories are periodic with the same period $n_i$.
Moreover, $n_i$ is a multiple of $d_i$ and every such a trajectory bounces the same number $\dfrac{n_i}{d_i}$ of times off each arc $\ell_{j}^i$;

\item
if billiard trajectories having vertices in $\delta_i$ are non-periodic, then the bouncing points of each trajectory are dense in $\delta_i$.
\end{itemize}
The boundary of each $\delta_i$ consists of bouncing points of saddle-connections.
\end{theorem}

This theorem is a consequence of Theorem \ref{th:maier} from the previous section.
The proof follows from the fact that each of the domains $D_i$ intersects the boundary $\Gamma_{\lambda_0}$ and forms $\delta_i=\Gamma_{\lambda_0}\cap D_i$.

This theorem is a variant of Maier theorem (see \cites{Maier1943,Viana}), i.e.~ Theorem \ref{th:poncelet} from the previous section.

In \cite{DragRadn2004} conditions of Cayley's type for periodicity of billiards within several confocal quadrics in the Euclidean space of an arbitrary dimension were derived, see also \cite{DragRadn2006b} for detailed examples.

We analyzed there billiards within domains bounded by arcs of several confocal quadrics and \emph{the billiad ordered game} within a few confocal ellipsoids.
Unlike in the present article, domains considered in \cite{DragRadn2004,DragRadn2006b} did not contain reflex angles at the boundary.
However, the technique used there to describe periodic trajectories can be directly transferred to the present problems.

Before stating the Cayley-type conditions, recall that a point is being reflected off conic $\mathcal{C}_{\lambda_0}$ \emph{from outside} if the corresponding Jacobi elliptic coordinate achieves a local maximum at the reflection point, and \emph{from inside} if there the coordinate achieves a local minimum (see \cite{DragRadn2004}).

\begin{theorem}\label{th:cayley}
Consider domain $\mathcal{D}$ bounded by arcs of $k$ ellipses $\mathcal{C}_{\beta_1}$, \dots, $\mathcal{C}_{\beta_k}$,  $l$ hyperbolas $\mathcal{C}_{\gamma_1}$, \dots, $\mathcal{C}_{\gamma_l}$, and several segments belonging to degenerate conics from the confocal family (\ref{eq:confocal}):
$$
\beta_1,\dots,\beta_k\in(-\infty,b),\ k\ge1,\ 
\gamma_1,\dots,\gamma_l\in(b,a),\ l\ge0.
$$
Let $\mathcal{C}_{\alpha_0}$ be an ellipse contained within all ellipses $\mathcal{C}_{\beta_1}$, \dots, $\mathcal{C}_{\beta_k}$: $b>\alpha_0>\beta_i$ for all $i\in\{1,\dots,k\}$.
A necessary condition for the existence of a billiard trajectory within $\mathcal{D}$ with $\mathcal{C}_{\alpha_0}$ as a caustic which becomes closed after:
\begin{itemize}
\item 
$n_i'$ reflections from inside and $n_i''$ reflections from outside off $\mathcal{C}_{\beta_i}$, $1\le i\le k$;

\item
$m_j'$ reflections from inside and $m_j''$ reflections from outside off $\mathcal{C}_{\gamma_j}$, $1\le j\le l$;
 
\item
total number of $p$ intersections with the $x$-axis and reflections off the segments contained in the $x$-axis;

\item
total number of $q$ intersections with the $y$-axis and reflections off the segments contained in the $y$-axis;
\end{itemize}
is:
\begin{gather*}
\sum_{i=1}^k(n_i'-n_i'')(\mathcal{A}(P_{\beta_i})-\mathcal{A}(P_{\alpha_0})) +
\sum_{j=1}^l(m_j'-m_j'')\mathcal{A}(P_{\gamma_j}) +
p\mathcal{A}(P_a)-q\mathcal{A}(P_b)=0,
\\
m_j'-m_j''+p-q=0.
\end{gather*}
Here $\mathcal{A}$ is the Abel-Jacobi map of the ellitic curve:
$$
\Gamma\ :\ s^2=\mathcal{P}(t):=(a-t)(b-t)(\alpha_0-t),
$$
and $P_{\delta}$ denotes point $(\delta,\sqrt{\mathcal{P}(\delta)})$ on $\Gamma$.
\end{theorem}

\begin{proof}
Following Jacobi \cite{JacobiGW} and Darboux \cite{Darboux1870}, similarly as in \cite{DragRadn2004}, we consider sums
$$
\int\frac{d\lambda_1}{\sqrt{\mathcal{P}(\lambda_1)}}+\int\frac{d\lambda_2}{\sqrt{\mathcal{P}(\lambda_2)}}
\ \text{and}\ 
\int\frac{\lambda_1d\lambda_1}{\sqrt{\mathcal{P}(\lambda_1)}}+\int\frac{\lambda_2d\lambda_2}{\sqrt{\mathcal{P}(\lambda_2)}}
$$
over billiard trajectory $A_1\dots A_N$.
Here $(\lambda_1,\lambda_2)$ are Jacobi elliptic coordinates, $\lambda_1<\lambda_2$.
The second integral is equal to the length of the trajectory, while the first one is zero.

Notice that, along a trajectory, $\lambda_1$ achieves local extrema at points of reflection off ellipses and touching points with the caustic, and $\lambda_2$ at points of reflection off hyperbolas and intersection points with the coordinate axes, we obtain that $A_1=A_N$ is equivalent to the condition stated.
\end{proof}

We illustrate this theorem on the example when the billiard table is $\mathcal{D}_0$, as defined in Section \ref{sec:uvod-nosonja}.

\begin{example}
A necessary condition for the existence of a billiard trajectory within $\mathcal{D}_0$ with $\mathcal{C}_{\alpha_0}$ as a caustic, such that it becomes closed after $n_1$ reflections off $\mathcal{C}_{\beta_1}$ and $n_2$ reflections off $\mathcal{C}_{\beta_2}$ is:
$$
n_1\mathcal{A}(P_{\beta_1})+n_2\mathcal{A}(P_{\beta_2})=(n_1+n_2)\mathcal{A}(P_{\alpha_0}).
$$
Notice that in this case number $p$ and $q$ are always even and equal to each other.
Since $2\mathcal{A}(P_a)=2\mathcal{A}(P_b)$, the corresponding summands are cancelled out.
\end{example}

\section{Connection with interval exchage transformation}\label{sec:interval}
In this section, we are going to establish a connection of the billiard dynamics within domain $\mathcal{D}_0$ defined in Section \ref{sec:uvod-nosonja} with the inteval exchange transformation.

\subsection{Interval exchange maps}\label{sec:int.ex}
Let $I\subset\mathbf{R}$ be an interval, and $\{I_{\alpha}\mid\alpha\in\mathcal{A}\}$ its finite partition into subintervals.
Here $\mathcal{A}$ is a finite set of at least two elements.
We consider all intervals to be closed on the left and open on the right.

\emph{An interval exchange map} is a bijection of $I$ into itself, such that its restriction on each $I_{\alpha}$ is a translation.
Such a map $f$ is determined by the following data:
\begin{itemize}
 \item
a pair $(\pi_0,\pi_1)$ of bijections $\mathcal{A}\to\{1,\dots,d\}$ describing the order of the subintervals $\{I_{\alpha}\}$ in $I$ and $\{f(I_{\alpha})\}$ in $f(I)=I$.
We denote:
$$
\pi=\left(
\begin{array}{cccc}
\pi_0^{-1}(1) & \pi_0^{-1}(2) & \dots & \pi_0^{-1}(d)\\
\pi_1^{-1}(1) & \pi_1^{-1}(2) & \dots & \pi_1^{-1}(d)
\end{array}
\right).
$$

\item
a vector $\lambda=(\lambda_{\alpha})_{\alpha\in\mathcal{A}}$ of the lengths of $I_{\alpha}$.
\end{itemize}

\subsection{Billiard dynamics}\label{sec:dynamics}
To each billiard trajectory, we join the sequence:
$$
\{(X_n,s_n)\},
\quad
X_n\in\mathcal{C}_{\alpha_0},
\quad
s_n\in\{+,-\}
$$
where $X_n$ are joint points of the trajectory with the caustic,
while $s_n=+$ if at $X_n$ the trajectory is winding counterclockwise and $s_n=-$ if it is winding clockwise about the caustic.

Introduce metric $\mu$ on the caustic $\mathcal{C}_{\alpha_0}$ as in Proposition \ref{prop:rotation}.
Then, we parametrize $\mathcal{C}_{\alpha_0}$ by parameters:
$$
p\ :\ \mathcal{C}_{\alpha_0}\to[0,1),
\quad
q\ :\ \mathcal{C}_{\alpha_0}\to[-1,0),
$$
which are natural with respect to $\mu$ such that $p$ is oriented counterclockwise and $q$ clockwise along $\mathcal{C}_{\alpha_0}$, and the values $p=0$ and $q=-1$ correspond to points $P_0$, $Q_0$ respectively, as shown in Figure \ref{fig:param}.
\begin{figure}[h]
\begin{pspicture}(-6,-2.5)(7,2.5)
  \SpecialCoor

%prva elipsa

\psset{linecolor=blue,linewidth=1pt}

\psline(!-3 2 sqrt)(!2 45 cos mul 3 sub 45 sin)
\psline(!2 135 cos mul 3 sub 135 sin)(!1.21218 3 sub 2.02031)
(!2 11.6015 cos mul 3 sub 11.6015 sin)
\psline(!-3 2 sqrt neg)(!2 45 cos mul neg 3 sub 45 sin neg)

\psset{linecolor=black,linewidth=1.3pt}

\psellipticarc(-3,0)(!8 sqrt 5 sqrt){-90}{90}
\psellipticarc(-3,0)(!5 sqrt 2 sqrt){90}{270}
\psellipse(-3,0)(2, 1)

\psline(!-3 2 sqrt)(!-3 5 sqrt)
\psline(!-3 2 sqrt neg)(!-3 5 sqrt neg)

\psset{linecolor=blue}

\pscircle*(!2 45 cos mul 3 sub 45 sin){0.08}
\pscircle*(!2 135 cos mul 3 sub 135 sin){0.08}
\pscircle*(!2 11.6015 cos mul 3 sub 11.6015 sin){0.08}
\pscircle*(!2 45 cos mul neg 3 sub 45 sin neg){0.08}

%\psset{linecolor=red}
%\pscircle*(-5,0){0.08}
%\pscircle*(-3,1){0.08}

\psset{linecolor=black}

\rput(!2 11.6015 cos mul 0.2 add 3 sub 11.6015 sin 0.2 add){$P_0$}
\rput(!2 45 cos mul 0.2 sub 3 sub 45 sin 0.2 sub){$P_1$}
\rput(!2 45 cos mul neg 0.2 add 3 sub 45 sin neg 0.2 add){$P_2$}

%\rput(-5.6,0){$\mathcal{C}_{\lambda_2}$}
%\rput(.2,0){$\mathcal{C}_{\lambda_1}$}

%druga elipsa

\psset{linecolor=blue,linewidth=1pt}

\psline(!3 2 sqrt neg)(!2 45 cos mul 3 add 45 sin neg)
\psline(!2 135 cos mul 3 add 135 sin neg)(!1.21218 3 add 2.02031 neg)
(!2 11.6015 cos mul 3 add 11.6015 sin neg)
\psline(!3 2 sqrt)(!2 45 cos mul neg 3 add 45 sin)

\psset{linecolor=black,linewidth=1.3pt}

\psellipticarc(3,0)(!8 sqrt 5 sqrt){-90}{90}
\psellipticarc(3,0)(!5 sqrt 2 sqrt){90}{270}
\psellipse(3,0)(2, 1)

\psline(!3 2 sqrt)(!3 5 sqrt)
\psline(!3 2 sqrt neg)(!3 5 sqrt neg)

\psset{linecolor=blue}

\pscircle*(!2 45 cos mul 3 add 45 sin neg){0.08}
\pscircle*(!2 135 cos mul 3 add 135 sin neg){0.08}
\pscircle*(!2 11.6015 cos mul 3 add 11.6015 sin neg){0.08}
\pscircle*(!2 45 cos mul neg 3 add 45 sin neg neg){0.08}

%\psset{linecolor=red}
%\pscircle*(-5,0){0.08}
%\pscircle*(-3,1){0.08}

\psset{linecolor=black}

\rput(!2 11.6015 cos mul 0.3 add 3 add 11.6015 sin neg 0.2 sub){$Q_0$}
\rput(!2 45 cos mul 0.2 sub 3 add 45 sin neg 0.3 add){$Q_1$}
\rput(!2 45 cos mul neg 0.2 add 3 add 45 sin neg neg 0.3 sub){$Q_2$}

\end{pspicture}
\caption{Parametrizations of the caustic.}\label{fig:param}
\end{figure}
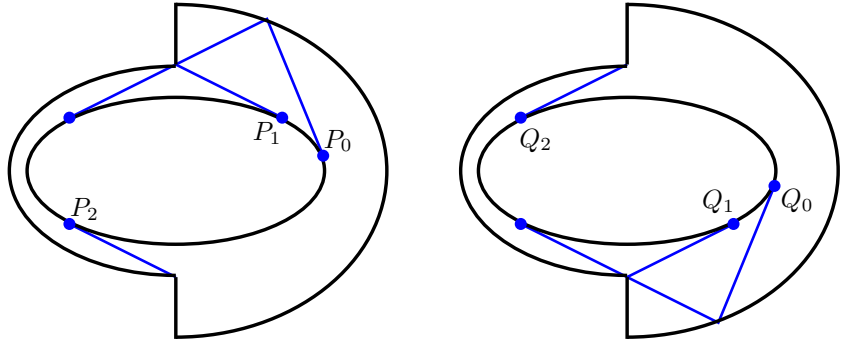

Consider one segment of a billiard trajectory, and let $X\in\mathcal{C}_{\alpha_0}$ be its touching point with the caustic.
Suppose that the particle is moving counterclockwise on that segment.
From Figure \ref{fig:param}, we conclude:
\begin{itemize}
\item
if $X$ is between points $P_1$ and $P_2$ then the particle is going to hit the arc
$\mathcal{C}_{\lambda_2}$;

\item
if $X$ is between $P_2$ and $P_0$, the particle is going to hit the arc
$\mathcal{C}_{\lambda_1}$;

\item
for $X$ between $P_0$ and $P_1$, the particle is going to hit $\mathcal{C}_{\lambda_1}$ and the upper segment before the next contact with the caustic and the direction of motion is changed to clockwise.
\end{itemize}
Similarly, if the particle is moving in clockwise direction, we have:
\begin{itemize}
\item
if $X$ is between points $Q_1$ and $Q_2$ then the particle is going to hit the arc
$\mathcal{C}_{\lambda_2}$;

\item
if $X$ is between $Q_2$ and $Q_0$, the particle is going to hit the arc
$\mathcal{C}_{\lambda_1}$;

\item
for $X$ between $Q_0$ and $Q_1$, the particle is going to hit $\mathcal{C}_{\lambda_1}$ and the lower segment before the next contact with the caustic and the direction of motion is changed to counterclockwise.
\end{itemize}

To see the billiard dynamics as an inteval exchange transformation, we make the following identification:
$$
(X,+)\sim p(X),
\quad
(X,-)\sim q(X).
$$
In other words:
\begin{itemize}
\item
we identify the joint point $X$ of a given trajectory with the caustic with 
$p(X)\in[0,1)$ if the particle is moving in the counterclockwise direction on the corresponding segment;

\item
for the motion in the clockwise direction, we identify $X$ with $q(X)\in[-1,0)$.
\end{itemize}

Denote the rotation numbers $r_1=\rho(\lambda_1)$, $r_2=\rho(\lambda_2)$ (see Proposition \ref{prop:rotation}).

The parametrizations values for points denoted in Figure \ref{fig:param} are:
\begin{gather*}
p(P_0)=0,
\quad
p(P_1)=r_1-r_2,
\quad
p(P_2)=r_1-r_2+\dfrac12,
\\
q(Q_0)=-1,
\quad
q(Q_1)=r_1-r_2-1,
\quad
q(Q_2)=r_1-r_2-\dfrac12.
\end{gather*}

Now, we distinguish three cases depending on the position of point $P_0$ with respect to the $x$-axis (see Figure \ref{fig:param}), i.e.~ on the sign of 
$\dfrac14+\dfrac{r_2}2-r_1$.

\subsection*{$P_0$ is on the $x$-axis: $\dfrac14+\dfrac{r_2}2-r_1=0$}
The interval exchange map is:
$$
\xi\mapsto
\begin{cases}
\xi+r_1+\frac32, &\xi\in[-1,-\frac12-r_1)\\
\xi+r_2, &\xi\in[-\frac12-r_1,-r_1)\\
\xi+r_1-1, &\xi\in[-r_1,0)\\
\xi+r_1-\frac12, & \xi\in[0,\frac12-r_1)\\
\xi+r_2, &\xi\in[\frac12-r_1,1-r_1)\\
\xi+r_1-1, &\xi\in[1-r_1,1),
\end{cases}
$$
as shown in Figure \ref{fig:interval0}.
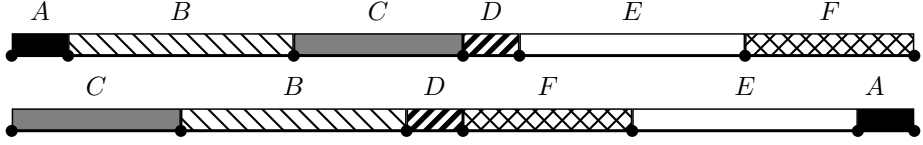
\begin{figure}[h]
\begin{pspicture}(-1.5,-1)(1.5,1)
\psset{xunit=6}

%gornji interval

\psset{linecolor=black,linewidth=0.02,fillstyle=solid}

\psframe[fillcolor=black](-1,0.5)(-0.875,0.8)
\psframe[fillstyle=vlines](-0.875,0.5)(-0.375,0.8)
\psframe[fillcolor=gray](-0.375,0.5)(0,0.8)
\psframe[fillstyle=hlines,hatchwidth=2pt,hatchsep=2pt](0,0.5)(0.125,0.8)
\psframe[fillcolor=white](0.125,0.5)(0.625,0.8)
\psframe[fillstyle=crosshatch](0.625,0.5)(1,0.8)

\psset{linecolor=black,linewidth=1pt}

\psline(-1,0.5)(1,0.5)

\pscircle*(-1,0.5){0.08}
\pscircle*(-0.875,0.5){0.08}
\pscircle*(-0.375,0.5){0.08}
\pscircle*(0,0.5){0.08}
\pscircle*(0.125,0.5){0.08}
\pscircle*(0.625,0.5){0.08}
\pscircle*(1,0.5){0.08}

%\rput(-1.07,0.5){$-1$}
%\rput(-0.875,0.2){$-\frac12-r_1$}
%\rput(-0.375,0.2){$-r_1$}
%\rput(0,0.2){$0$}
%\rput(0.125,0.2){$\frac12-r_1$}
%\rput(0.625,0.2){$1-r_1$}
%\rput(1.04,0.5){$1$}

\rput(-0.9375,1.1){$A$}
\rput(-0.625,1.1){$B$}
\rput(-0.1875,1.1){$C$}
\rput(0.0625,1.1){$D$}
\rput(0.375,1.1){$E$}
\rput(0.8125,1.1){$F$}

%donji interval

\psset{linecolor=black,linewidth=0.02,fillstyle=solid}

\psframe[fillcolor=gray](-1,-.5)(-0.625,-.2)
\psframe[fillstyle=vlines](-0.625,-.5)(-0.125,-.2)
\psframe[fillstyle=hlines,hatchwidth=2pt,hatchsep=2pt](-0.125,-.5)(0,-.2)
\psframe[fillstyle=crosshatch](0,-.5)(0.375,-.2)
\psframe[fillcolor=white](0.375,-.5)(0.875,-.2)
\psframe[fillcolor=black](0.875,-.5)(1,-0.2)

\psset{linecolor=black,linewidth=1pt}

\psline(-1,-.5)(1,-.5)

\pscircle*(-1,-0.5){0.08}
\pscircle*(-0.625,-0.5){0.08}
\pscircle*(-0.125,-0.5){0.08}
\pscircle*(0,-0.5){0.08}
\pscircle*(0.375,-0.5){0.08}
\pscircle*(0.875,-0.5){0.08}
\pscircle*(1,-0.5){0.08}

\rput(-0.8125,.1){$C$}
\rput(-0.375,.1){$B$}
\rput(-0.0625,.1){$D$}
\rput(0.1875,.1){$F$}
\rput(0.625,.1){$E$}
\rput(0.9125,.1){$A$}

\end{pspicture}
\caption{Interval exchange transformation for the case $\frac14+\frac{r_2}2-r_1=0$.}\label{fig:interval0}
\end{figure}

To the map, pair $(\pi,\lambda)$ is joined:
\begin{gather*}
\pi=\left(
\begin{array}{cccccc}
A & B & C & D & E & F \\
C & B & D & F & E & A 
\end{array}
\right),
\\
\lambda=\left(\frac12-r_1,\ \frac12,\ r_1,\ \frac12-r_1,\ \frac12,\ r_1\right).
\end{gather*}

\subsection*{$P_0$ is above the $x$-axis: $\dfrac14+\dfrac{r_2}2-r_1>0$}
The interval exchange map in this case is shown in Figure \ref{fig:interval+} and given by:
$$
\xi\mapsto
\begin{cases}
\xi+r_1+\frac32, &\xi\in[-1,r_1-r_2-1)\\
\xi+r_2, &\xi\in[r_1-r_2-1,r_1-r_2-\frac12)\\
\xi+r_1, &\xi\in[r_1-r_2-\frac12,-r_1)\\
\xi+r_1-1, &\xi\in[-r_1,0)\\

\xi+r_1-\frac12, &\xi\in[0,r_1-r_2)\\
\xi+r_2, &\xi\in[r_1-r_2,r_1-r_2+\frac12)\\
\xi+r_1, &\xi\in[r_1-r_2+\frac12,1-r_1)\\
\xi+r_1-1, &\xi\in[1-r_1,1).
\end{cases}
$$
\begin{figure}[h]
\begin{pspicture}(-1.5,-1)(1.5,1)
\psset{xunit=6}

%gornji interval

\psset{linecolor=black,linewidth=0.02,fillstyle=solid}

\psframe[fillcolor=black](-1,0.5)(-0.9,0.8)
\psframe[fillstyle=vlines](-0.9,0.5)(-0.4,0.8)
\psframe[fillcolor=gray](-0.4,0.5)(-0.35,0.8)
\psframe[fillstyle=hlines,hatchwidth=2pt,hatchsep=2pt](-0.35,0.5)(0,0.8)
\psframe[fillcolor=white](0,0.5)(0.1,0.8)
\psframe[fillstyle=crosshatch](0.1,0.5)(0.6,0.8)
\psframe[fillcolor=gray!30](0.6,0.5)(0.65,0.8)
\psframe[fillstyle=hlines,hatchangle=0,hatchsep=1.5pt](0.65,0.5)(1,0.8)

\psset{linecolor=black,linewidth=1pt}

\psline(-1,0.5)(1,0.5)

\pscircle*(-1,0.5){0.08}
\pscircle*(-0.9,0.5){0.08}
\pscircle*(-0.4,0.5){0.08}
\pscircle*(-0.35,0.5){0.08}
\pscircle*(0,0.5){0.08}
\pscircle*(0.1,0.5){0.08}
\pscircle*(0.6,0.5){0.08}
\pscircle*(0.65,0.5){0.08}
\pscircle*(1,0.5){0.08}

\rput(-0.95,1.1){$A$}
\rput(-0.65,1.1){$B$}
\rput(-0.375,1.1){$C$}
\rput(-0.175,1.1){$D$}
\rput(0.05,1.1){$E$}
\rput(0.35,1.1){$F$}
\rput(0.625,1.1){$G$}
\rput(0.825,1.1){$H$}

%donji interval

\psset{linecolor=black,linewidth=0.02,fillstyle=solid}

\psframe[fillstyle=hlines,hatchwidth=2pt,hatchsep=2pt](-1,-.5)(-0.65,-0.2)
\psframe[fillstyle=vlines](-0.65,-.5)(-0.15,-0.2)
\psframe[fillcolor=white](-0.15,-.5)(-0.05,-0.2)
\psframe[fillcolor=gray](-0.05,-.5)(0,-0.2)
\psframe[fillstyle=hlines,hatchangle=0,hatchsep=1.5pt](0,-.5)(0.35,-0.2)
\psframe[fillstyle=crosshatch](0.35,-.5)(0.85,-0.2)
\psframe[fillcolor=black](0.85,-.5)(0.95,-0.2)
\psframe[fillcolor=gray!30](0.95,-.5)(1,-0.2)

\psset{linecolor=black,linewidth=1pt}

\psline(-1,-.5)(1,-.5)

\pscircle*(-1,-.5){0.08}
\pscircle*(-0.65,-.5){0.08}
\pscircle*(-0.15,-.5){0.08}
\pscircle*(-0.05,-.5){0.08}
\pscircle*(0,-.5){0.08}
\pscircle*(0.35,-.5){0.08}
\pscircle*(0.85,-.5){0.08}
\pscircle*(0.95,-.5){0.08}
\pscircle*(1,-.5){0.08}

\rput(-0.825,.1){$D$}
\rput(-0.4,.1){$B$}
\rput(-0.1,.1){$E$}
\rput(-0.025,.1){$C$}
\rput(0.175,.1){$H$}
\rput(0.6,.1){$F$}
\rput(0.9,.1){$A$}
\rput(0.975,.1){$G$}

\end{pspicture}
\caption{Interval exchange transformation for the case $\frac14+\frac{r_2}2-r_1>0$.}\label{fig:interval+}
\end{figure}

The map can be desribed by the pair $(\pi,\lambda)$:
\begin{gather*}
\pi=\left(
\begin{array}{cccccccc}
A & B & C & D & E & F & G & H\\
D & B & E & C & H & F & A & G 
\end{array}
\right),
\\
\lambda=\left(r_1-r_2,\ \frac12,\ r_2-2r_1+\frac12,\ r_1,\ r_1-r_2,\ \frac12,\ r_2-2r_1+\frac12,\ r_1\right).
\end{gather*}

\subsection*{$P_0$ is below the $x$-axis: $\dfrac14+\dfrac{r_2}2-r_1<0$}
The interval exchange map corresponding to the billiard dynamics is:
$$
\xi\mapsto
\begin{cases}
\xi+r_1+\frac32, &\xi\in[-1,-\frac12-r_1)\\
\xi+r_1-\frac12, &\xi\in[-\frac12-r_1,r_1-r_2-1)\\
\xi+r_2, &\xi\in[r_1-r_2-1,r_1-r_2-\frac12)\\
\xi+r_1-1, &\xi\in[r_1-r_2-\frac12,0)\\

\xi+r_1-\frac12, &\xi\in[0,\frac12-r_1)\\
\xi+r_1-\frac32, &\xi\in[\frac12-r_1,r_1-r_2)\\
\xi+r_2, &\xi\in[r_1-r_2,r_1-r_2+\frac12)\\
\xi+r_1-1, &\xi\in[r_1-r_2+\frac12,1),
\end{cases}
$$
see Figure \ref{fig:interval-}.
\begin{figure}[h]
\begin{pspicture}(-1.5,-1)(1.5,1)
\psset{xunit=6}

%gornji interval

\psset{linecolor=black,linewidth=0.02,fillstyle=solid}

\psframe[fillcolor=black](-1,0.5)(-0.9,0.8)
\psframe[fillstyle=vlines](-0.9,0.5)(-0.75,0.8)
\psframe[fillcolor=gray](-0.75,0.5)(-0.25,0.8)
\psframe[fillstyle=hlines,hatchwidth=2pt,hatchsep=2pt](-0.25,0.5)(0,0.8)
\psframe[fillcolor=white](0,0.5)(0.1,0.8)
\psframe[fillstyle=crosshatch](0.1,0.5)(0.25,0.8)
\psframe[fillcolor=gray!30](0.25,0.5)(0.75,0.8)
\psframe[fillstyle=hlines,hatchangle=0,hatchsep=1.5pt](0.75,0.5)(1,0.8)

\psset{linecolor=black,linewidth=1pt}

\psline(-1,0.5)(1,0.5)

\pscircle*(-1,0.5){0.08}
\pscircle*(-0.9,0.5){0.08}
\pscircle*(-0.75,0.5){0.08}
\pscircle*(-0.25,0.5){0.08}
\pscircle*(0,0.5){0.08}
\pscircle*(0.1,0.5){0.08}
\pscircle*(0.25,0.5){0.08}
\pscircle*(0.75,0.5){0.08}
\pscircle*(1,0.5){0.08}

\rput(-0.95,1.1){$A$}
\rput(-0.825,1.1){$B$}
\rput(-0.5,1.1){$C$}
\rput(-0.125,1.1){$D$}
\rput(0.05,1.1){$E$}
\rput(0.175,1.1){$F$}
\rput(0.5,1.1){$G$}
\rput(0.825,1.1){$H$}

%donji interval

\psset{linecolor=black,linewidth=0.02,fillstyle=solid}

\psframe[fillstyle=crosshatch](-1,-0.5)(-0.85,-0.2)
\psframe[fillstyle=hlines,hatchwidth=2pt,hatchsep=2pt](-0.85,-0.5)(-0.6,-0.2)
\psframe[fillcolor=gray](-0.6,-0.5)(-0.1,-0.2)
\psframe[fillcolor=white](-0.1,-0.5)(0,-0.2)
\psframe[fillstyle=vlines](0,-0.5)(0.15,-0.2)
\psframe[fillstyle=hlines,hatchangle=0,hatchsep=1.5pt](0.15,-0.5)(0.4,-0.2)
\psframe[fillcolor=gray!30](0.4,-0.5)(0.9,-0.2)
\psframe[fillcolor=black](0.9,-0.5)(1,-0.2)

\psset{linecolor=black,linewidth=1pt}

\psline(-1,-.5)(1,-.5)

\pscircle*(-1,-.5){0.08}
\pscircle*(-0.85,-.5){0.08}
\pscircle*(-0.6,-.5){0.08}
\pscircle*(-0.1,-.5){0.08}
\pscircle*(0,-.5){0.08}
\pscircle*(0.15,-.5){0.08}
\pscircle*(0.4,-.5){0.08}
\pscircle*(0.9,-.5){0.08}
\pscircle*(1,-.5){0.08}

\rput(-0.925,.1){$F$}
\rput(-0.725,.1){$D$}
\rput(-0.35,.1){$C$}
\rput(-0.05,.1){$E$}
\rput(0.075,.1){$B$}
\rput(0.275,.1){$H$}
\rput(0.65,.1){$G$}
\rput(0.95,.1){$A$}

\end{pspicture}
\caption{Interval exchange transformation for the case $\frac14+\frac{r_2}2-r_1<0$.}\label{fig:interval-}
\end{figure}

To the map, pair $(\pi,\lambda)$ is joined:
\begin{gather*}
\pi=\left(
\begin{array}{cccccccc}
A & B & C & D & E & F & G & H\\
F & D & C & E & B & H & G & A 
\end{array}
\right),
\\
\lambda=\left(\frac12-r_1,\ 2r_1-r_2-\frac12,\ \frac12, r_2+\frac12-r_1,\ \frac12-r_1,\ 2r_1-r_2-\frac12,\ \frac12, r_2+\frac12-r_1\right).
\end{gather*}

Notice that in all three cases the interval exchange transformations depend only on the rotation numbers $r_1$, $r_2$.
Thus, we got

\begin{theorem}\label{th:nezavisnost}
The billiard dynamics inside the domain $\mathcal{D}_0$ with ellipse $\mathcal{C}_{\alpha_0}$ as the caustic, does not depend on the parameters $a$, $b$ of the confocal family but only on the rotation numbers $r_1$, $r_2$.
\end{theorem}

\section{Keane condition and minimality}\label{sec:keane}
An interval exchange transformation is called \emph{minimal} if every orbit is dense in the whole domain.
When considering pseudo-billiards, minimal interval exchange transformations will correspond to the cases when all orbits are dense in the domain between the billiard border and the caustic.

Following \cite{Viana}, we are going to formulate a sufficient condition for minimality.
Let $f$ be an interval exchange transformation of $I$, given by pair $(\pi,\lambda)$.
Denote by $p_{\alpha}$ the left endpoint of $I_{\alpha}$.
Then the transformation satisfies \emph{the Keane condition} if:
$$
f^m(p_{\alpha})\neq p_{\beta}
\ \text{for all}\ m\ge1,\ \alpha\in\mathcal{A},\ \beta\in\mathcal{A}\setminus\{\pi_0^{-1}(1)\}. 
$$
Obviously, none of the transformations from Section \ref{sec:interval} satisfies the Keane condition:
namely, the midpoint of the interval is the left endpoint of one of $I_{\alpha}$, and it is the image of another endpoint in the corresponding interval exchange map.

The goal of this section is to find an analogue of the Keane condition for interval exhange transformations appearing in the billiard dynamics.

\subsection{Billiard-like transformations and modified Keane condition}
Analysis of the examples from Section \ref{sec:interval} motivates the following definitions.

\begin{definition}\label{def:billiard-like}
An interval exchange transformation $f$ of $I=[-1,1)$ is \emph{billiard-like} if the partition into subintervals satisfies the following:
\begin{itemize}
 \item
for each $\alpha$, $I_{\alpha}$ is contained either in $[-1,0)$ or $[0,1)$;

 \item
for each $\alpha$, $f(I_{\alpha})$ is contained either in $[-1,0)$ or $[0,1)$;

 \item
both $[-1,0)$ and $[0,1)$ contain at least two intervals of the partition.
\end{itemize}
\end{definition}

\begin{definition}\label{def:keane}
We will say that a billiard-like interval exchange transformation $f$ satisfies \emph{the modified Keane condition} if
$$
f^m(p_{\alpha})\neq p_{\beta}
\ \text{for all}\ m\ge1,\ \alpha\in\mathcal{A},
\ \text{and}\ \beta\in\mathcal{B}\ \text{such that}\ p_{\beta}\not\in\{-1,0\}. 
$$
\end{definition}

\begin{lemma}\label{lema:keane.periodic}
If a billiard-like interval exchange transformation satisfies the modified Keane condition, then the transformation has no periodic points.
\end{lemma}

\begin{proof}
Suppose the transformation has a periodic point.
Then there is $\alpha\in\mathcal{A}$ such that the left endpoint of $I_{\alpha}$ is periodic (\cite{Viana}), i.e.~  $f^m(p_{\alpha})=p_{\alpha}$ for some $m\ge1$.

By the modified Keane condition $p_{\alpha}\in\{-1,0\}$.
Without losing generality, take $p_{\alpha}=-1$.

Hence we got $f^m(-1)=-1$.
If $m=1$, take $I_{\beta}$ to be the interval adjacent to $I_{\alpha}$.
Notice that $-1<p_{\beta}<0$.
Then $p_{\beta}=f(p_{\gamma})$ for some $\gamma$, which contradicts the modified Keane condition.

Now take $m>1$.
The point $p_{\beta}=f^{-1}(-1)$ is also periodic with period $m$, thus $p_{\beta}=0$, i.e.~ $f(0)=-1$ and $f^m(0)=0$.
Analogously, $f(-1)=0$.

If intervals $I_{\alpha}$ and $I_{\beta}$ are of the same length, then the left endpoints of their adjacent intervals are images of some left endpoints as well, which contradicts the modified Keane condition.
Thus, suppose that $I_{\alpha}$ is shorter than $I_{\beta}$: $\lambda_{\alpha}<\lambda_{\beta}$.
Point $\lambda_{\alpha}\in I$ is the right endpoint of $f(I_{\alpha})$, thus it is the image of a left endpoint of some interval: $\lambda_{\alpha}=f(p_{\gamma})$.
Since $\lambda_{\alpha}\in I_{\beta}$, $f(\lambda_{\alpha})$ is the left endpoint of the interval $I_{\delta}$ adjacent to $I_{\alpha}$.
Thus, $f^2(p_{\gamma})=p_{\delta}$ and $p_{\delta}\not\in\{-1,0\}$, which contradicts the modified Keane condition.
\end{proof}

We say that an interval exchange transformation is \emph{irreducible} if for no $k<|\mathcal{A}|$ the union
$$
I_{\alpha_{\pi_0^{-1}(1)}}\cup\dots\cup I_{\alpha_{\pi_0^{-1}(k)}}
$$
is invariant under the transformation.
The usual Keane condition implies irreducibility.
However, this is not the case for the modified Keane condition -- it may happen that the transformation falls apart into two irreducible transformations on $[-1,0)$ and $[0,1)$.
On the other hand, if for a transformation satisfying the modified Keane condition there is an interval $I_{\alpha}\subset[-1,0)$ such that $f(I_{\alpha})\subset[0,1)$, the irreducibility will also take place.

\begin{proposition}\label{prop:keane.min}
If an irreducible billiard-like interval exchange transformation $f$ satisfies the modified Keane condition, then $f$ is minimal.
\end{proposition}

\begin{proof}
Let $x\in I$ be a point whose orbit is not dense in $I$.
Then there is an interval $J\subset I$ which is disjoint with the orbit of $x$.
Moreover, we can choose $J$ such that it is entirely contained in $I_{\alpha}$ for some $\alpha\in\mathcal{A}$.

It is shown in \cite{Viana} that the first return map of $f$ to $J$ is an interval exchange transformation.
As a consequence, the union $\hat{J}$ of all orbits of points of $J$ is a finite union of intervals and a fully invariant set: $f(\hat{J})=\hat{J}$ (\cite{Viana}).

Moreover, $\hat{J}\neq I$, because $\hat{J}$ is also disjoint with the orbit of $x$.

\emph{First step: we prove that $\hat{J}$ contains a connected component with the left endpoint not in $\{-1,0\}$.}

Suppose the opposite -- that $\hat{J}$ is an interval with the left endpoint equal to $-1$ or $0$, or the union of two such intervals.
If any of the connected components of $\hat{J}$ would be contained in one of the intervals $I_{\alpha}$, then $f\mid I_{\alpha}$ would be the identity map, which leads to a contradiction with the modified Keane condition.
Thus $\hat{J}$ contains some of the intervals of the partition -- let $I_{\alpha_1}$, \dots, $I_{\alpha_k}$ be all of them.
In this case, $I_{\alpha_1}\cup\dots\cup I_{\alpha_k}$ is invariant under the transformation.
If this union, or its connected component, is of the form $[-1,a)$, $a\neq0$, or $[0,b)$, $b\neq1$, this will be in the contradiction with the modified Keane condition; if it is $[-1,0)$ or $[0,1)$ --- the irreducibility property is violated; it cannot concide with the whole interval $I$ because it is disjoint with the orbit of $x$.

We conclude that not all connected components of $\hat{J}$ can be intervals with the left endpoints in $-1$ or $0$.

\emph{Second step: we prove that all left endpoints of connected components of $\hat{J}$ are also left endpoints of the partition intervals.}

Suppose now that $[a,b)$ is a connected component of $\hat{J}$, while $a$ not being a left endpoint of an interval of the partition.
Since $f$ is continuous at inner points of the partition intervals, and $\hat{J}$ is fully invariant, it follows that $f(a)$ is also the left endpoint of some connected component of $\hat{J}$.
If none of the points $f^n(a)$, $n>0$ is a left endpoint of an interval of the partition,
by induction we get that each of these points is on the boundary of some connected component of $\hat{J}$.
There are finitely many such components, thus $a$ is a periodic point, which is not possible by Lemma \ref{lema:keane.periodic}.
Hence, there is $n>0$ such that $f^n(a)$ is a left endpoint of an interval of the partition.
For the smallest such $n$, $p_{\alpha}=f^n(a)\not\in\{-1,0\}$, since $f^{n-1}(a)$ is an inner point of a partition interval.

Similarly, we find $m>0$ such that $f^{-m}(a)=p_{\beta}$ for some $\beta\in\mathcal{A}$.
Now the relation $f^{m+n}(p_{\beta})=p_{\alpha}$ contradicts the modified Keane condition.

\emph{Third step: consider the complement of $\hat{J}$.}

Set $\hat{J}^c=I\setminus\hat{J}$ is a fully invariant nonempty set.
Thus we can prove the same what we proved for $\hat{J}$ -- each connected component of the set has it left endpoint at a left endpoint of an interval of the partition and at least of these endpoints is neither $-1$ nor $0$.

Thus both $\hat{J}$ and $\hat{J}^c$ are fully invariant sets that can be represented as the unions of intervals of the partition.
Since they contain connected components with left endpoints not in $\{-1,0\}$ this leads to a contradiction with the modified Keane condition.

The final contradiction leads us to the conclusion that the initial assumption of the existence of a non-dense orbit was not valid.
\end{proof}

\subsection{An example}
Consider billiard trajectories within domain $\mathcal{D}_0$ with the caustic $\mathcal{C}_{\alpha_0}$, as described in Section \ref{sec:uvod-nosonja}.
In addition, suppose the rotation numbers corresponding to ellipses $\mathcal{C}_{\lambda_1}$ and $\mathcal{C}_{\lambda_2}$ are:
$$
r_1=\frac5{11}+\frac1{22\pi},\quad r_2=\frac5{11}-\frac1{220\pi}.
$$

With given rotation numbers, the Cayley-type conditions from Theorem \ref{th:cayley} can be rewritten in a simpler form.
Namely, a necessary condition for existence of a trajectory within $\mathcal{D}_0$ which becomes closed after $n$ reflections of $\mathcal{C}_{\lambda_1}$ and $m$ reflections off $\mathcal{C}_{\lambda_2}$ is:
$$
nr_1+mr_2\in\mathbf{Z}.
$$
In this case, this condition is satisfied for $n=1$ and $m=10$:
\begin{equation}\label{eq:r1+10r2}
r_1+10r_2=5.
\end{equation}

Since $\dfrac14+\dfrac{r_2}2-r_1>0$, the corresponding interval exhange transformation is given by:
\begin{gather*}
\Pi=\left(
\begin{array}{cccccccc}
A & B & C & D & E & F & G & H\\
D & B & E & C & H & F & A & G 
\end{array}
\right),
\\
\lambda=\left(\frac1{20\pi},\ \frac12,\ \frac1{22}-\frac{21}{220\pi},\ \frac5{11}+\frac1{22\pi},
\ \frac1{20\pi},\ \frac12,\ \frac1{22}-\frac{21}{220\pi},\ \frac5{11}+\frac1{22\pi}\right).
\end{gather*}

\begin{proposition}\label{prop:keane-example}
The transformation $(\Pi,\lambda)$ satisfies the modified Keane condition.
\end{proposition}

\begin{proof}
Suppose that $p$ and $p'$ are two endpoints of the intervals such that $p'\not\in\{-1,0\}$ and $f^k(p)=p'$ for some $k\ge1$.
Notice that:
$$
p=\alpha r_1+\beta r_2+\gamma\frac12,\quad
p'=\alpha' r_1+\beta' r_2+\gamma'\frac12,
$$
for some
$\alpha,\alpha'\in\{-1,0,1\}$, $\beta,\beta'\in\{-1,0\}$, $\gamma,\gamma'\in\{-2,-1,0,1,2\}$.

We have:
$$
p'=f^k(p)=p+k_1r_1+k_2r_2+k_3\frac12,
$$
for some integers $k_1$, $k_2$, $k_3$ such that $k_1+k_2=k$, $k_1\ge0$, $k_2\ge0$.
Thus:
\begin{equation}\label{eq:k1r1+k2r2}
(k_1+\alpha-\alpha')r_1+(k_2+\beta-\beta')r_2+(k_3+\gamma-\gamma')\frac12=0.
\end{equation}
Since $r_1$ and $r_2$ are irrational, equations (\ref{eq:r1+10r2}) and (\ref{eq:k1r1+k2r2}) must be dependent:
\begin{equation}\label{eq:a}
a:=k_1+\alpha-\alpha'=\frac1{10}(k_2+\beta-\beta')=-\frac1{10}(k_3+\gamma-\gamma').
\end{equation}
For each $\xi\in B\cup F$, either $f(\xi)$ or $f^2(\xi)$ are not in $B\cup F$, thus
\begin{equation}\label{eq:relk1k2}
k_2\le2k_1+2.
\end{equation}
Combining (\ref{eq:relk1k2}) and (\ref{eq:a}) we get $8a\le7$.
Since $k_2$ is non-negative, (\ref{eq:relk1k2}) gives that $a=0$, which leads to $k=k_1+k_2\le3$.
By direct calculation we check that none of the partition interval endpoints is mapped into another one, different from $-1$ and $0$ by at most three iterations.
\end{proof}

In this example, although the Cayley-type conditon for periodicity is satisfied, not only that closed trajectories do not exist, but each of the trajectories densely fills the ring between the billiard border and the caustic.

\end{document}